\documentclass[11pt]{amsart}
\usepackage{amssymb,mathrsfs,graphicx,enumerate}
\usepackage{amsmath,amsfonts,amssymb,amscd,amsthm,bbm}
\usepackage[retainorgcmds]{IEEEtrantools}
\usepackage{colortbl}

\title[Complete aggregation of the Lohe tensor model]{Complete aggregation of the Lohe tensor model with the same free flow}

\author[Ha]{Seung-Yeal Ha}
\address[Seung-Yeal Ha]{\newline Department of Mathematical Sciences\newline Seoul National University, Seoul 08826, and \newline
Korea Institue for Advanced Study, Hoegiro 85, 02455, Seoul, Republic of Korea}
\email{syha@snu.ac.kr}

\author[Park]{Hansol Park}
\address[Hansol Park]{\newline Department of Mathematical Sciences\newline Seoul National University, Seoul 08826, Republic of Korea}
\email{hansol960612@snu.ac.kr}

\newtheorem{theorem}{Theorem}[section]
\newtheorem{lemma}{Lemma}[section]

\newtheorem{proposition}{Proposition}[section]
\newtheorem{remark}{Remark}[section]

\newtheorem{definition}{Definition}[section]

\newcommand{\bbr}{\mathbb R}

\newcommand{\bbc}{\mathbb C}

\begin{document}

\date{\today}


\thanks{\textbf{Acknowledgment.} The work of S.-Y.Ha is supported by NRF-2020R1A2C3A01003881}

\begin{abstract}
The Lohe tensor model is a first-order tensor-valued continuous-time model for the aggregation of tensors with the same rank and size. It reduces to well-known aggregation models such as the Kuramoto model, the Lohe sphere model and the Lohe matrix model as special cases for low-rank tensors. We present a sufficient and necessary framework for the solution splitting property(SSP) and analyze two possible asymptotic states(completely aggregate state and bi-polar state) which can emerge from a set of initial data. Moreover, we provide a sufficient framework leading to the aforementioned two asymptotic states in terms of initial data and system parameters. 
\end{abstract}
\maketitle \centerline{\date}


\section{Introduction} \label{sec:1}
\setcounter{equation}{0}
The purpose of this paper is to continue a systematic study \cite{H-P1, H-P2, H-P3} on the aggregation of Lohe tensor flock. The Lohe tensor model encompass well known Lohe type aggregation models such as the Kuramoto model, the Lohe sphere model and the Lohe matrix model. Mathematical modeling and analysis for collective dynamics were initiated by two pioneers A. Winfree and Y. Kuramoto \cite{Ku1, Ku2, Wi2, Wi1} in a half century ago. Due to the increasing demand from engineering community, it has been one of hot topics in applied mathematics, control theory and statistical physics, etc. After Winfree's seminal work, many mathematicl models have been proposed in related communities (see survey papers \cite{A-B, A-B-F, D-B1, H-K-P-Z, P-R, St, VZ}). Among them, we are mainly interested in the Lohe tensor model introduced by the authors in \cite{H-P3} which unifies aforementioned Lohe type models. To fix the idea, we begin with a brief introduction on tensors \cite{B-C, Or} and Lohe tensor model. 

A {\it tensor} presents a multi-dimensional array of complex numbers with several indices, and it can be regarded as a generalization of vector and matrix, and the rank of a tensor is the number of indices, i.e., a rank-$m$ tensor of dimensions $d_1 \times \cdots \times d_m$ is an element of ${\mathbb C}^{d_1 \times \cdots \times d_m}$. Hence rank-$m$ tensor $T \in {\mathbb C}^{d_1 \times \cdots \times d_m}$ can be identified as a ${\mathbb C}$-valued multilinear map from ${\mathbb C}^{d_1} \times \cdots {\mathbb C}^{d_m}$ to ${\mathbb C}$. It is easy to see that scalars, vectors, and matrices correspond to rank-0, 1, and 2 tensors, respectively. For a rank-$m$ tensor $T$, $\alpha_* := (\alpha_1, \cdots, \alpha_m) \in \{1, \cdots, d_1 \} \times \cdots \times \{1, \cdots, d_m \}$-th component of $T$ is given by $[T]_{\alpha_*} = [T]_{\alpha_1 \cdots \alpha_m}$, and ${\bar T}$ denotes the rank-$m$ tensor whose components are the complex conjugate of the elements of $T$:
\[ [{\bar T}]_{\alpha_1 \cdots \alpha_m} =\overline{[T]_{\alpha_1 \cdots \alpha_m}}. \]
A set ${\mathcal T}_m(\bbc) := {\mathcal T}_m(d_1, \cdots, d_m;{\mathbb C})$ is a complex vector space consisting of all rank-$m$ tensors with complex entries and the size $d_1 \times \cdots \times d_m$. The addition and scalar multiplication can be defined componentwise. One key operation in ${\mathcal T}_m({\mathbb C})$ is an index contraction by using Einstein summation rule. For example, inner product between matrices with the same size can be defined as a contraction between rank-2 tensors. For $T \in {\mathcal T}_m(\bbc)$, we set 
\begin{align*}
\begin{aligned}
& [T]_{\alpha_{*}}:=[T]_{\alpha_{1}\alpha_{2}\cdots\alpha_{m}}, \quad  [T]_{\alpha_{*0}}:=[T]_{\alpha_{10}\alpha_{20}\cdots\alpha_{m0}}, \quad [T]_{\alpha_{*1}}:=[T]_{\alpha_{11}\alpha_{21}\cdots\alpha_{m1}}, \\
&[T]_{\alpha_{*i_*}}:=[T]_{\alpha_{1i_1}\alpha_{2i_2}\cdots\alpha_{mi_m}}, \quad  [T]_{\alpha_{*(1-i_*)}}:=[T]_{\alpha_{1(1-i_1)}\alpha_{2(1-i_2)}\cdots\alpha_{m(1-i_m)}},
\end{aligned}
\end{align*}
Additionally, we also define Frobenius inner product, its induced norm and an ensemble diameter as follows: for an ensemble $\{T_i \}$, 
\[  \langle T_i, T_j \rangle_F:= [\bar{T}_i]_{\alpha_{*0}} [T_j]_{\alpha_{*0}}, \quad \| T_i \|_F^2 :=  \langle T_i, T_i \rangle_F  \quad \mbox{and} \quad {\mathcal D}(T) := \max_{1 \leq i,j \leq N} \|T_i - T_j \|_F . \]
Here we used Einstein summation rule for repeated indices. Now we are ready to describe the Lohe tensor model in componentwise. \newline

Suppose that the dynamics of $[T_j]_{\alpha_{*0}}$ is governed by the Cauchy problem:
\begin{equation} 
\begin{cases} \label{A-1}
\displaystyle\frac{d}{dt}[T_j]_{\alpha_{*0}} = \underbrace{[A_j]_{\alpha_{*0}\beta_*}[T_j]_{\beta_*}}_{\mbox{free flow}} \\
\hspace{1.5cm} + \underbrace{\sum_{i_* \in \{0, 1\}^m}{\kappa_{i_*}}([T_c]_{\alpha_{*i_*}}[\bar{T}_j]_{\alpha_{*1}}[T_j]_{\alpha_{*(1-i_*)}}-[T_j]_{\alpha_{*i_*}}[\bar{T}_c]_{\alpha_{*1}}[T_j]_{\alpha_{*(1-i_*)}})}_{\mbox{cubic interactions}}, \\
[T_j]_{\alpha_{*0}} \Big|_{t = 0+} = [T^{in}_j]_{\alpha_{*0}}, 
\end{cases}
\end{equation}
where $A_j$ is a special type of rank-$2m$ tensor satisfying the following property:
\begin{equation} \label{A-2}
 [A]_{\alpha_*\beta_*}:=[A]_{\alpha_{1}\alpha_{2}\cdots\alpha_{m}\beta_1\beta_2\cdots\beta_{m}}, \quad [A_j]_{\alpha_*\beta_*}=-[\bar{A_j}]_{\beta_*\alpha_*}.
\end{equation}
In a series of works \cite{H-P1, H-P2, H-P3}, the authors studied emergent dynamics of system \eqref{A-1} - \eqref{A-2} and its reductions to low-rank tensor models such as the Kuramoto model, the Lohe sphere model and Lohe matrix model  (see Section \ref{sec:2.2} for details). Although system \eqref{A-1} looks so complicate, the Frobenius norm of $T_j$ is conserved along the Lohe tensor flow \eqref{A-1} - \eqref{A-2}:
\[ \|T_j(t) \|_F = \|T_j(0) \|_F, \quad t \geq 0, \quad j = 1, \cdots, N. \]
Next, we briefly review the results in \cite{H-P3}. For a homogeneous ensemble with the same free flow, if the coupling strength $\kappa_0$ and the initial data $\{T_j^{in} \}$ satisfy
\begin{equation} \label{A-3}
 \kappa_{0 \cdots 0} \gg \sum_{i_*\neq0}\kappa_{i_*} \quad \mbox{and} \quad {\mathcal D}({\mathcal T}^{in}) \ll 1,
\end{equation}
then the ensemble diameter decays exponentially fast:
\[ \lim_{t \to \infty} {\mathcal D}(T(t))  = 0. \]
In contrast, for a heterogeneous ensemble with distributed $\{A_j \}$, if the natural frequency tensors and initial data satisfy 
\[  {\mathcal D}(T^{in}) \ll 1  \quad \mbox{and} \quad {\mathcal D}({\mathcal A}):= \max_{1 \leq i,j \leq N} \|A_i - A_j \|_F  \ll 1, \]
then practical synchronization occurs asymptotically:
\[  \lim_{{\mathcal D}(A)/\kappa_0 \to 0+}  \limsup_{t\rightarrow\infty}{\mathcal D}(T)=0.  \]
In this paper, we focus on the more detailed emergent dynamics for a homogeneous Lohe tensor flock with the same free flow, more precisely, we are interested in the following two questions: \newline
\begin{itemize}
\item
Can we relax the conditions \eqref{A-3} on coupling strengths and initial data?
\vspace{0.2cm}
\item
What are the phase-locked states emerging from relaxation process?
\end{itemize}

\vspace{0.5cm}

The main results of this paper are two-fold. First, we present a necessary and sufficient framework for the solution splitting property(SSP) to \eqref{A-1} - \eqref{A-2}. Here the solution splitting property means that the solution to the full nonlinear system \eqref{A-1} with the same free flow can be expressed as a composition of the nonlinear flow and corresponding free flow. If SSP holds, we can assume $A = 0$ without loss of generality. Thus, we focus on the nonlinear part. Let $A$ and $T$ be rank-$2m$ and rank-$m$ tensors satisfying the following three properties: for $n \in {\mathbb Z}_+$:
\begin{align}
\begin{aligned} \label{A-3-0}
&[A^0]_{\alpha_*\beta_*}=\delta_{\alpha_*\beta_*}=\prod_{k=1}^m\delta_{\alpha_k\beta_k}, \quad [A^n]_{\alpha_*\beta_*}= \underbrace{[A]_{\alpha_*\gamma_{1*}}[A]_{\gamma_{1*}\gamma_{2*}}\cdots[A]_{\gamma_{(n-1)*}\beta_{*}}}_{n~\mbox{times of $A$}}, \\
&\mbox{and}\quad  [AT]_{\alpha_*}=[A]_{\alpha_*\beta_*}[T]_{\beta_*},
\end{aligned}
\end{align}
where $\delta_{\alpha_* \beta_*}$ is the rank-$2m$ Kronecker delta type function:
\[  \delta_{\alpha_* \beta_*} = \begin{cases}
1, \quad &\alpha_k = \beta_k, \quad \forall k = 1, \cdots, m, \\
0, \quad &\mbox{otherwise}.
\end{cases}      \]
Once rank-$2m$ tensor $A$ satisfies the relations \eqref{A-3-0}, we can define an exponential of $A$ similar to the matrix exponential:
\begin{equation} \label{A-3-1}
 [e^A]_{\alpha_*\beta_*}=\sum_{n=0}^{\infty}{1\over{n!}}[A^n]_{\alpha_*\beta_*}.   
 \end{equation}
Then, a sufficient and necessary condition for $A$ to satisfy SSP can be stated by the explicit relation:
\begin{equation} \label{A-4}
 [e^{-At}]_{\alpha_{*0}\beta_{*0}}[e^{At}]_{\beta_{*i_*}\gamma_{*i_*}}[e^{-{A}t}]_{\alpha_{*1}\beta_{*1}}[e^{At}]_{\beta_{*(1-i_*)}\gamma_{*(1-i_*)}} = \delta_{\alpha_{*0}\gamma_{*0}}\delta_{\alpha_{*1}\gamma_{*1}}.
\end{equation}
Under the condition \eqref{A-4}, solution to \eqref{A-1}  - \eqref{A-2} is given by the composition of the corresponding linear free flow and nonlinear flow (Theorem \ref{T2.1}):
\begin{align*}
&T_j=e^{At}S_j, \quad j  =1, \cdots, N, \\
&\dot{[S_j]}_{\alpha_{*0}} = \sum_{i_1, ... , i_m\in\{0, 1\}}\kappa_{i_*}\Big( [S_c]_{\alpha_{*i_*}}\bar{[S_j]}_{\alpha_{*1}}[S_j]_{\alpha_{*(1-i_*)}}-[S_j]_{\alpha_{*i_*}}\bar{[S_c]}_{\alpha_{*1}}[S_j]_{\alpha_{*(1-i_*)}} \Big).
\end{align*} 
Second, we present a sufficient framework for the complete aggregation. If coupling strengths $\kappa_{i_*}$ satisfy
\begin{equation} \label{A-4-1}
\kappa_{00\cdots0}>0,\quad \mbox{and}\quad\kappa_{i_*}\geq0 \quad  \forall i_* \neq (0,\cdots, 0).
\end{equation}
then there exists a vector ${\bf a} = (a_2, \cdots, a_N) \in \{-1, 1 \}^{N-1}$ such that 
\begin{equation} \label{D-12-0}
\lim_{t \to \infty}(T_i(t) - a_i T_1(t)) = 0, \quad \forall~ i = 2, \cdots, N.
\end{equation}
This leads to the dichotomy (complete aggregation to a singleton $\{T_1(t) \}$ or aggregation to bi-polar configuration $\{T_1(t), -T_1(t) \}$  for emergent dynamics to the Lohe tensor model with the same free flow. Under the conditions \eqref{A-4-1} on the coupling strengths, the homogeneous Lohe tensor flow leads to either completely aggregate state or bi-polar state asymptotically depending on initial data (see Theorem \ref{T4.1}).  In addition to \eqref{A-4-1} and $N \geq 3$, if initial data satisfy
\[  ||T_c(0)||_F>1-\frac{2}{N}, \]
then the homogeneous Lohe tensor flow tends to the completely aggregate state (see Theorem \ref{T4.2}).  \newline

The rest of this paper is organized as follows. In Section \ref{sec:2}, we briefly discuss basic properties of the Lohe tensor model and present a sufficient and necessary framework for the solution splitting property, and we also provide a brief comparison between the previous result and our results in this paper. In Section \ref{sec:3}, we present reshaping of tensor contractions as matrix multiplications, and introduce a fundamental Lohe tensor model consisting of a single interaction pair. In Section \ref{sec:4}, we introduce a variance-like functional and study its time rate of changes along the homogeneous Lohe tensor flow. By using the monotonicity of order parameter, we present a dichotomy (one-point concentration or bi-polar concentration) in the large-time dynamics to the homogeneous Lohe tensor flow. Finally, Section \ref{sec:5} is devoted to a brief summary of our main results and a future work to be discussed in future. 
\section{Preliminaries} \label{sec:2}
\setcounter{equation}{0}
In this section, we recall basic properties of the Lohe tensor model, solution splitting property and discuss low-rank tensor models such as the Kuramoto model, the Lohe sphere model and the Lohe matrix model and present a comparison with previous results. 

\subsection{The Lohe tensor model} \label{sec:2.1}
Let $\{T_j \}$ be the ensemble of the rank-$m$ Lohe tensor flock. Then, its dynamics is given by the following coupled system of ordinary differential equations:
\begin{equation}  \label{B-1}
\begin{cases}
\displaystyle \frac{d}{dt}[T_j]_{\alpha_{*0}} =[A_j]_{\alpha_{*0}\beta_*}[T_j]_{\beta_*} \\
\displaystyle \hspace{1.5cm} +\sum_{i_* \in \{0,1\}^m}{\kappa_{i_*}}([T_c]_{\alpha_{*i_*}}[\bar{T}_j]_{\alpha_{*1}}[T_j]_{\alpha_{*(1-i_*)}}-[T_j]_{\alpha_{*i_*}}[\bar{T}_c]_{\alpha_{*1}}[T_j]_{\alpha_{*(1-i_*)}}), \\
\displaystyle [A_j]_{\alpha_*\beta_*}=-[\bar{A_j}]_{\beta_*\alpha_*}.
\end{cases}
\end{equation}
\begin{lemma} \label{L2.1}
\emph{\cite{H-P3}}
Let $\{ T_j \}$ be a solution to the Lohe-tensor model \eqref{B-1}. Then $||T_j ||^2_F$ is a conserved quantity: 
\[  \frac{d}{dt}  ||T_j||^2_F = 0, \quad t > 0, \quad j = 1, \cdots, N. \]
\end{lemma}
\begin{proof} The proof can be found in Proposition 4.1 \cite{H-P3}. Thus, we omit its proof here.
\end{proof}
Next, we consider the Cauchy problem for the corresponding linear free flow system to \eqref{B-1}:
\begin{equation} \label{B-2}
\begin{cases}
\displaystyle \frac{dT}{dt} = AT, \quad t > 0, \\
\displaystyle T |_{t = 0} =T^{in},
\end{cases}
\end{equation}
where $A$ is rank-$2m$ tensor satisfying the suitable conditions below. Similar to matrix exponential, we define the exponential of constant tensor.
\begin{definition} \label{D2.1}
Let $A$ and $T$ be rank-$2m$ and rank-$m$ tensors with the following three properties: for $n \in {\mathbb Z}_+$,
\begin{align} 
\begin{aligned}\label{B-3}
&\displaystyle  [A^0]_{\alpha_*\beta_*}=\delta_{\alpha_*\beta_*}=\prod_{k=1}^m\delta_{\alpha_k\beta_k},\quad [A^n]_{\alpha_*\beta_*}=\underbrace{[A]_{\alpha_*\gamma_{1*}}[A]_{\gamma_{1*}\gamma_{2*}}\cdots[A]_{\gamma_{(n-1)*}\beta_{*}}}_{\mbox{$n$ times of A}},\\
&\mbox{and}\quad [AT]_{\alpha_*}=[A]_{\alpha_*\beta_*}[T]_{\beta_*}. 
\end{aligned}
\end{align}
Then, the exponential of $A$ is defined in \eqref{A-3-1}.  
\end{definition}
\begin{remark}
It is easy to see that $A^n(T)=A^{n-1}(AT)$.
\end{remark}
\begin{lemma} \label{L2.1} Let $A$ be a rank-$2m$ tensor satisfying the properties $\eqref{B-1}_2$ and \eqref{B-3}. Then, the exponential of tensor defined in \eqref{A-3-1} satisfies
\[
[e^{-A}]_{\alpha_*\beta_*}=[e^{\bar{A}}]_{\beta_*\alpha_*}, \quad [e^{At}]_{\alpha_*\beta_*}[e^{-At}]_{\beta_*\gamma_*}=\delta_{\alpha_*\gamma_*}.
\]
\end{lemma}
\begin{proof}
(i)~We use the relations $\eqref{B-1}_2$ and \eqref{B-3} to get 
\[
[e^{-A}]_{\alpha_*\beta_*} =\sum_{n=0}^{\infty} \frac{1}{n !} (-1)^n[A^n]_{\alpha_*\beta_*} =\sum_{n=0}^{\infty} \frac{1}{n !} [\bar{A}^n]_{\beta_*\alpha_*}=[e^{\bar{A}}]_{\beta_*\alpha_*}.
\]
(ii)~We use \eqref{B-3} to obtain
\begin{align*}
[e^{At}]_{\alpha_*\beta_*}[e^{-At}]_{\beta_*\gamma_*} &=\sum_{n=0}^{\infty}\sum_{m=0}^{\infty}\frac{1}{n!}\frac{1}{m!}[A^n]_{\alpha_*\beta_*}[(-A)^m]_{\beta_*\gamma_*} =\sum_{k=0}^{\infty}\sum_{n+m=k}\frac{(-1)^{m}}{n!m!}[A^{n+m}]_{\alpha_*\gamma_*}\\
&=\sum_{k=0}^{\infty}\sum_{n+m=k}\frac{(-1)^{m}}{n!m!}[A^{k}]_{\alpha_*\gamma_*} =\sum_{k=0}^{\infty}\left([A^{k}]_{\alpha_*\gamma_*}\sum_{n+m=k}\frac{(-1)^{m}}{n!m!}\right)\\
&=[A^0]_{\alpha_*\gamma_*}=\delta_{\alpha_*\gamma_*},
\end{align*}
where we used the following identity: for $k \geq 1$,
\[ 0 = (1 -1)^k = \sum_{n+ m = k} \frac{(-1)^m}{n! m !}. \]
\end{proof}

\vspace{0.2cm}

\begin{proposition} \label{P2.1}
The unique solution of the Cauchy problem \eqref{B-2} is explicitly given by 
\begin{equation} \label{B-5}
 T(t)=e^{At}T^{in}, \quad t \geq 0.
 \end{equation}
\end{proposition}
\begin{proof}  
\noindent Step A~(The formula \eqref{B-5} satisfies the Cauchy problem):  Since $e^{0 A} = I$, one has
\[ T(0) = e^{0 A} T^{in}  = T^{in}. \]
We use \eqref{B-5} to get
\begin{align*}
\begin{aligned}
\frac{d}{dt} T(t)&={d\over{dt}}e^{At}T^{in} ={d\over{dt}}\sum_{n=0}^{\infty}{1\over{n!}}(At)^nT^{in} = {d\over{dt}}\sum_{n=1}^{\infty}{1\over{n!}}(At)^nT^{in} =\sum_{n=1}^{\infty} \frac{A}{(n-1)!} (At)^{n-1}T^{in} \\
 &=  A \sum_{n=1}^{\infty}{1\over{(n-1)!}}(At)^{n-1}T^{in} =A\sum_{n=0}^{\infty}{1\over{n!}}(At)^{n}T^{0}=AT(t).
\end{aligned}
\end{align*}

\vspace{0.5cm}

\noindent Step B~(Uniqueness):~Let $S = S(t)$ be another solution to \eqref{B-2}. Then,
\[ \frac{d}{dt} \Big( e^{-t A}  S(t) \Big ) = -A e^{-tA} S(t) + e^{-tA} {\dot S}(t) = A e^{-tA} S(t)  + e^{-tA}  A S(t) = 0. \]
This yields
\[ e^{-t A}  S(t) = T^{in}, \quad t \geq 0, \quad \mbox{i.e.,} \quad S(t) = e^{t A} T^{in}. \] 
\end{proof}
\subsection{Solution splitting property (SSP)} \label{sec:2.2} 
In this subsection, we study sufficient and necessary conditions for the solution splitting property of  the Lohe tensor model which generalize the result of \cite{H-P1} on rank-2 tensor model. For this, we consider corresponding linear and nonlinear systems associated with \eqref{B-1}:
\begin{align*}
\begin{aligned}  \label{B-6}
& \dot{[T_j]}_{\alpha_{*0}} =[A]_{\alpha_{*0}\alpha_{*1}}[T_j]_{\alpha_{*1}} \quad \mbox{and} \\
& \dot{[S_j]}_{\alpha_{*0}} = \sum_{i_*\in\{0, 1\}^m}\kappa_{i_*}\Big( [S_c]_{\alpha_{*i_*}}\bar{[S_j]}_{\alpha_{*1}}[S_j]_{\alpha_{*(1-i_*)}}-[S_j]_{\alpha_{*i_*}}\bar{[S_c]}_{\alpha_{*1}}[S_j]_{\alpha_{*(1-i_*)}} \Big).
\end{aligned}
\end{align*}
Recall that our strategy is to find conditions on $A$ such that 
\begin{equation} \label{B-7}
T_j=e^{At}S_j, \quad j = 1, \cdots, N.
\end{equation}

Now we will define the class of the natural frequency tensor $A$ admitting the solution splitting property \eqref{B-7}. We define $\mathcal{C}(i_*)$ as follows:
\begin{align}
\begin{aligned} \label{B-7-0}
\mathcal{C}(i_*)&=\{A:\delta_{\alpha_{*0}\gamma_{*0}}\delta_{\gamma_{*1}\alpha_{*1}}=[e^{-At}]_{\alpha_{*0}\beta_{*0}}[e^{At}]_{\beta_{*i_*}\gamma_{*i_*}}[e^{-{A}t}]_{\alpha_{*1}\beta_{*1}}[e^{At}]_{\beta_{*(1-i_*)}\gamma_{*(1-i_*)}}\},\\
\mathcal{I}&=\{i_*: \kappa_{i_*}\neq0\}.
\end{aligned}
\end{align}
Then, the above discussion can be summarized in the following theorem.
\begin{theorem}\label{T2.1}
Suppose that rank-$2m$ tensor $A$ satisfies
\begin{align*}
 A\in\bigcap_{i_*\in\mathcal{I}}\mathcal{C}(i_*),
\end{align*}
and let $\{T_j \}$ be a solution to system \eqref{B-1}. Then, the solution splitting property holds:

\begin{align}\label{B-7-1}
\begin{aligned}
&T_j=e^{At}S_j \quad \mbox{and} \\
&\dot{[S_j]}_{\alpha_{*0}} = \sum_{i_*\in\{0, 1\}^m}\kappa_{i_*}\Big( [S_c]_{\alpha_{*i_*}}\bar{[S_j]}_{\alpha_{*1}} [S_j]_{\alpha_{*(1-i_*)}}-[S_j]_{\alpha_{*i_*}}\bar{[S_c]}_{\alpha_{*1}}[S_j]_{\alpha_{*(1-i_*)}} \Big).
\end{aligned}
\end{align}

\end{theorem}
\begin{proof}
We substitute the ansatz \eqref{B-7} into \eqref{B-1} to see
\begin{align*}
\begin{aligned}
&{d\over{dt}}\left([e^{At}]_{\alpha_{*0}\beta_*}{[S_j]}_{\beta_*}\right)
=[A]_{\alpha_{*0}\eta_*}[e^{At}]_{\eta_*\beta_*}[S_j]_{\beta_*}\\
& \hspace{3.5cm} +\sum_{i_*\in\{0, 1\}^m}\kappa_{i_*}\Big([e^{At}]_{\alpha_{*i_*}\gamma_{*}}[S_c]_{\gamma_{*}}[e^{\bar{A}t}]_{\alpha_{*1}\delta_{*}}\bar{[S_j]}_{\delta_{*}}[e^{At}]_{\alpha_{*(1-i_*)}\epsilon_{*}}[S_j]_{\epsilon_{*}}\\
& \hspace{5cm}  -[e^{At}]_{\alpha_{*i_*}\gamma_{*}}[S_j]_{\gamma_{*}}[e^{\bar{A}t}]_{\alpha_{*1}\delta_{*}}\bar{[S_c]}_{\delta_{*}}[e^{At}]_{\alpha_{*(1-i_*)}\epsilon_{*}}[S_j]_{\epsilon_{*}} \Big).
\end{aligned}
\end{align*}
This and Proposition \ref{P2.1} imply
\begin{align}
\begin{aligned} \label{B-8}
[e^{At}]_{\alpha_{*0}\beta_*}\frac{d}{dt}{[S_j]}_{\beta_{*}}
&  =\sum_{i_*\in\{0, 1\}^m}\kappa_{i_*}\big([e^{At}]_{\alpha_{*i_*}\gamma_{*}}[S_c]_{\gamma_{*}}[e^{\bar{A}t}]_{\alpha_{*1}\delta_{*}}\bar{[S_j]}_{\delta_{*}}[e^{At}]_{\alpha_{*(1-i_*)}\epsilon_{*}}[S_j]_{\epsilon_{*}}\\
&\hspace{1cm}  -[e^{At}]_{\alpha_{*i_*}\gamma_{*}}[S_j]_{\gamma_{*}}[e^{\bar{A}t}]_{\alpha_{*1}\delta_{*}}\bar{[S_c]}_{\delta_{*}}[e^{At}]_{\alpha_{*(1-i_*)}\epsilon_{*}}[S_j]_{\epsilon_{*}}\big)\\
&=\sum_{i_*\in\{0, 1\}^m}\kappa_{i_*}[e^{At}]_{\alpha_{*i_*}\gamma_{*}}[e^{\bar{A}t}]_{\alpha_{*1}\delta_{*}}[e^{At}]_{\alpha_{*(1-i_*)}\epsilon_{*}}  \\  
& \hspace{1.5cm} \times \Big([S_c]_{\gamma_{*}}\bar{[S_j]}_{\delta_{*}}[S_j]_{\epsilon_{*}}
-[S_j]_{\gamma_{*}}\bar{[S_c]}_{\delta_{*}}[S_j]_{\epsilon_{*}} \Big).
\end{aligned}
\end{align}
We multiply  $[e^{-At}]_{\beta_* \alpha_{*0}}$ to the L.H.S. of \eqref{B-8} to get 
\begin{align*}
\begin{aligned} 
{d\over{dt}}{[S_j]}_{\beta_*} &=\sum_{i_* \in\{0, 1\}^m}\kappa_{i_*}[e^{-At}]_{\beta_*\alpha_{*0}}[e^{At}]_{\alpha_{*i_*}\gamma_{*}}[e^{-{A}t}]_{\delta_{*}\alpha_{*1}}[e^{At}]_{\alpha_{*(1-i_*)}\epsilon_{*}}  \\
& \hspace{1.5cm} \times \Big([S_c]_{\gamma_{*}}\bar{[S_j]}_{\delta_{*}}[S_j]_{\epsilon_{*}}
-[S_j]_{\gamma_{*}}\bar{[S_c]}_{\delta_{*}}[S_j]_{\epsilon_{*}}\Big).
\end{aligned}
\end{align*}
If we change dummy variables, we can obtain 
\begin{align}
\begin{aligned} \label{B-10}
{d\over{dt}}{[S_j]}_{\alpha_{*0}} &=\sum_{i_*\in\{0, 1\}^m}\kappa_{i_*}[e^{-At}]_{\alpha_{*0}\beta_{*0}}[e^{At}]_{\beta_{*i_*}\gamma_{*}}[e^{-{A}t}]_{\delta_{*}\beta_{*1}}[e^{At}]_{\beta_{*(1-i_*)}\epsilon_{*}}  \\
&\hspace{1.5cm} \times
\Big([S_c]_{\gamma_{*}}\bar{[S_j]}_{\delta_{*}}[S_j]_{\epsilon_{*}}
-[S_j]_{\gamma_{*}}\bar{[S_c]}_{\delta_{*}}[S_j]_{\epsilon_{*}}\Big).
\end{aligned}
\end{align}
Now, we compare $\eqref{B-7-1}_2$ and \eqref{B-10} to find admissible cubic couplings:
\begin{align*}
\begin{aligned}
&[S_c]_{\alpha_{*i_*}}\bar{[S_j]}_{\alpha_{*1}}[S_j]_{\alpha_{*(1-i_*)}}-[S_j]_{\alpha_{*i_*}}\bar{[S_c]}_{\alpha_{*1}}[S_j]_{\alpha_{*(1-i_*)}}\\
& \hspace{0.5cm} =[e^{-At}]_{\alpha_{*0}\beta_{*0}}[e^{At}]_{\beta_{*i_*}\gamma_{*}}[e^{-{A}t}]_{\delta_{*}\beta_{*1}}[e^{At}]_{\beta_{*(1-i_*)}\epsilon_{*}} 
\\
& \hspace{1cm} \times \Big([S_c]_{\gamma_{*}}\bar{[S_j]}_{\delta_{*}}[S_j]_{\epsilon_{*}}
-[S_j]_{\gamma_{*}}\bar{[S_c]}_{\delta_{*}}[S_j]_{\epsilon_{*}}\Big).
\end{aligned}
\end{align*}
Then we can easily transform above equation as follows:
\begin{align*}
&\left(\delta_{\alpha_{*i_*}\gamma_*}\delta_{\alpha_{*1}\delta_*}\delta_{\alpha_{*(1-i_*)}\epsilon_*}-[e^{-At}]_{\alpha_{*0}\beta_{*0}}[e^{At}]_{\beta_{*i_*}\gamma_{*}}[e^{-{A}t}]_{\delta_{*}\beta_{*1}}[e^{At}]_{\beta_{*(1-i_*)}\epsilon_{*}} \right)\\
&\hspace{3cm}\times\Big([S_c]_{\gamma_{*}}\bar{[S_j]}_{\delta_{*}}[S_j]_{\epsilon_{*}}
-[S_j]_{\gamma_{*}}\bar{[S_c]}_{\delta_{*}}[S_j]_{\epsilon_{*}}\Big)=0.
\end{align*}
Since the above relation should hold for every $\{S_j\}$, we can obtain
\begin{align}\label{B-11}
\delta_{\alpha_{*i_*}\gamma_*}\delta_{\alpha_{*1}\delta_*}\delta_{\alpha_{*(1-i_*)}\epsilon_*}=[e^{-At}]_{\alpha_{*0}\beta_{*0}}[e^{At}]_{\beta_{*i_*}\gamma_{*}}[e^{-{A}t}]_{\delta_{*}\beta_{*1}}[e^{At}]_{\beta_{*(1-i_*)}\epsilon_{*}}.
\end{align}
We continue to simplify the above condition \eqref{B-11} as follows. Since $\gamma_*$ and $\epsilon_*$ are dummy variables, we can set
\[   \gamma_*\rightarrow \gamma_{*i_*} \quad \mbox{and} \quad \epsilon_{*}\rightarrow\gamma_{*(1-i_*)}. \]
Then we have
\begin{equation} \label{B-11-1}
\delta_{\alpha_{*i_*}\gamma_{*i_*}}\delta_{\alpha_{*1}\delta_*}\delta_{\alpha_{*(1-i_*)}\gamma_{*(1-i_*)}}=[e^{-At}]_{\alpha_{*0}\beta_{*0}}[e^{At}]_{\beta_{*i_*}\gamma_{*i_*}}[e^{-{A}t}]_{\delta_{*}\beta_{*1}}[e^{At}]_{\beta_{*(1-i_*)}\gamma_{*(1-i_*)}}.
\end{equation}
The L.H.S of \eqref{B-11-1} can be expressed using the definition of generalized delta function:
\[
\delta_{\alpha_{*i_*}\gamma_{*i_*}}\delta_{\alpha_{*1}\delta_*}\delta_{\alpha_{*(1-i_*)}\gamma_{*(1-i_*)}}=\delta_{\alpha_{*1}\gamma_{*1}}\delta_{\alpha_{*0}\gamma_{*0}}\delta_{\alpha_{*1}\delta_*}=\delta_{\alpha_{*0}\gamma_{*0}}\delta_{\gamma_{*1}\delta_*}.
\]
Hence, we have
\[
\delta_{\alpha_{*0}\gamma_{*0}}\delta_{\gamma_{*1}\delta_*}=[e^{-At}]_{\alpha_{*0}\beta_{*0}}[e^{At}]_{\beta_{*i_*}\gamma_{*i_*}}[e^{-{A}t}]_{\delta_{*}\beta_{*1}}[e^{At}]_{\beta_{*(1-i_*)}\gamma_{*(1-i_*)}}.
\]
Since $\delta_*$ is dummy variable, we can put $\delta_*\rightarrow\alpha_{*1}$. Finally we obtain the desired consistency condition in \eqref{B-7-0} for 
the solution splitting property:
\begin{align}\label{B-12}
\delta_{\alpha_{*0}\gamma_{*0}}\delta_{\gamma_{*1}\alpha_{*1}}=[e^{-At}]_{\alpha_{*0}\beta_{*0}}[e^{At}]_{\beta_{*i_*}\gamma_{*i_*}}[e^{-{A}t}]_{\alpha_{*1}\beta_{*1}}[e^{At}]_{\beta_{*(1-i_*)}\gamma_{*(1-i_*)}}.
\end{align}
\end{proof}
Throughout the rest of this paper, we call the relation \eqref{B-12} as a consistency condition on $A$ for a solution splitting property with $\kappa_{i_*}$. Note that the consistency conditions will be checked for $2^m$ choices for the coupling strength $i_*$.  In the following subsection, we consider explicit aggregation models and check how the consistency conditions \eqref{B-12} can be reduced for those explicit models.

\subsection{SSP for low-rank tensors} Consider the Kuramoto model for identical oscillators \cite{B-C-M, C-H-J-K, C-S, D-X, D-B1, D-B, H-L-X, H-S}:
\begin{equation} \label{B-13}
\begin{cases}
\displaystyle {\dot \theta}_j = \nu + \frac{\kappa}{N} \sum_{k=1}^{N} \sin(\theta_k - \theta_j), \quad j  = 1, \cdots, N, \\
\displaystyle A = \nu \in \bbr
\end{cases}
\end{equation}
Then, the condition \eqref{B-12} is valid trivially:
\[ e^{-\nu t}\times e^{\nu t}\times e^{-\nu t}\times e^{\nu t} = 1 \times 1. \]
Hence, the Kuramoto model \eqref{B-13} satisfies the solution splitting property.  \newline

Next, we consider the Lohe sphere model with the same natural frequency matrix \cite{C-C-H, C-H5,M-T-G,T-M,Zhu}:
\begin{equation} \label{B-14}
\begin{cases}
\displaystyle {\dot x}_j = \Omega x_j + \frac{\kappa}{N} \sum_{k=1}^{N} \Big( \langle x_j, x_j \rangle x_k - \langle x_k, x_j \rangle x_j \Big), \quad j = 1, \cdots, N, \\
\displaystyle \Omega^{t} = -\Omega \in \bbr^{(d+1) \times (d+1)},
\end{cases}
\end{equation}
Then, it is easy to see that the condition \eqref{B-12} holds for \eqref{B-14}:
\[
[e^{-\Omega t}]_{\alpha_{10}\beta_{10}}[e^{\Omega t}]_{\beta_{10}\gamma_{10}}[e^{-{\Omega}t}]_{\alpha_{11}\beta_{11}}[e^{\Omega t}]_{\beta_{11}\gamma_{11}} = \delta_{\alpha_{10}\gamma_{10}}\delta_{\gamma_{11}\alpha_{11}}.
\]

\vspace{0.5cm}

Finally, we consider the Lohe matrix model on the unitary group ${\mathbb U}(d)$ \cite{H-K, H-K-R,H-R, Lo-0, Lo-1, Lo-2}:
\begin{equation} \label{B-15}
\begin{cases}
\displaystyle \dot{U}_j=-\mathrm{i}H U_j +\frac{\kappa}{2N}\sum_{k=1}^N(U_k-U_jU_k^\dagger U_j), \quad j = 1, \cdots, N, \\
\displaystyle  H^\dagger = H,
\end{cases}
\end{equation}
where $\dagger$-operation denotes hermitian conjugation. See \cite{B-C-S, D-F-M-T, D-F-M, De} for similar matrix aggregation models. To put \eqref{B-15} into the Lohe tensor model, we set rank-4 tensor $A$ as
\begin{equation*} \label{B-16}
[-\mathrm{i}H U_i]_{\alpha\beta}=[A]_{\alpha\beta\gamma\delta}[U_i]_{\gamma\delta}, \quad \mbox{i.e.,} \quad [A]_{\alpha\beta\gamma\delta}=[-\mathrm{i}H]_{\alpha\gamma}\delta_{\beta\delta}.
\end{equation*}
This yields 
\begin{equation} \label{B-17}
[A^n]_{\alpha\beta\gamma\delta}=[(-\mathrm{i}H)^n]_{\alpha\gamma}\delta_{\beta\delta}.
\end{equation}
Finally, we use \eqref{A-3-1} and \eqref{B-17} to get 
\[
[e^{At}]_{\alpha\beta\gamma\delta}=\sum_{n=0}^\infty \frac{t^n}{n!}[A^n]_{\alpha\beta\gamma\delta} = \sum_{n=0}^\infty \frac{t^n}{n!}[(-\mathrm{i}H)^n]_{\alpha\gamma}\delta_{\beta\delta}=[e^{-\mathrm{i}H}]_{\alpha\gamma}\delta_{\beta\delta}.
\]
Similarly, one has
\begin{equation*} \label{B-18}
[e^{-At}]_{\alpha\beta\gamma\delta}=[e^{\mathrm{i}H}]_{\alpha\gamma}\delta_{\beta\delta},
\end{equation*}
Consider the Lohe tensor model with $m = 2$:
\begin{equation*} \label{B-19}
\kappa_{{00}}=\kappa_{{10}}=\kappa_{{11}}=0, \quad \kappa_{{01}}=\frac{\kappa}{2}.
\end{equation*}
If we substitute this $A$ in \eqref{B-5} to R.H.S. of \eqref{A-4} with $i_*=(0, 1)$, then one has the condition \eqref{A-4}:
\begin{align*}
&[e^{-At}]_{\alpha_{10}\alpha_{20}\beta_{10}\beta_{20}}[e^{At}]_{\beta_{10}\beta_{21}\gamma_{10}\gamma_{21}}[e^{-{A}t}]_{\alpha_{11}\alpha_{21}\beta_{11}\beta_{21}}[e^{At}]_{\beta_{11}\beta_{20}\gamma_{11}\gamma_{20}}\\
& \hspace{1cm} =[e^{\mathrm{i}H}]_{\alpha_{10}\beta_{10}}\delta_{\alpha_{20}\beta_{20}}[e^{-\mathrm{i}H}]_{\beta_{10}\gamma_{10}}\delta_{\beta_{21}\gamma_{21}}[e^{\mathrm{i}H}]_{\alpha_{11}\beta_{11}}\delta_{\alpha_{21}\beta_{21}}[e^{-\mathrm{i}H}]_{\beta_{11}\gamma_{11}}\delta_{\beta_{20}\gamma_{20}}\\
& \hspace{1cm} =[e^{\mathrm{i}H}]_{\alpha_{10}\beta_{10}}[e^{-\mathrm{i}H}]_{\beta_{10}\gamma_{10}}[e^{\mathrm{i}H}]_{\alpha_{11}\beta_{11}}[e^{-\mathrm{i}H}]_{\beta_{11}\gamma_{11}}\delta_{\alpha_{20}\beta_{20}}\delta_{\beta_{21}\gamma_{21}}\delta_{\alpha_{21}\beta_{21}}\delta_{\beta_{20}\gamma_{20}}\\
& \hspace{1cm} =\delta_{\alpha_{10}\gamma_{10}}\delta_{\alpha_{11}\gamma_{11}}\delta_{\alpha_{20}\beta_{20}}\delta_{\beta_{21}\gamma_{21}}\delta_{\alpha_{21}\beta_{21}}\delta_{\beta_{20}\gamma_{20}}\\
& \hspace{1cm} =\delta_{\alpha_{10}\gamma_{10}}\delta_{\alpha_{20}\gamma_{20}}\delta_{\gamma_{11}\alpha_{11}}\delta_{\alpha_{21}\gamma_{21}}.
\end{align*}

\subsection{Comparison with previous result} \label{sec:2.4} 
In this subsection, we briefly present related earlier result on the emergent dynamics for \eqref{A-1}. For notational simplicity, we set 
\[ \kappa_0 := \kappa_{0 \cdots 0}. \]
Then, we first recall the result in \cite{H-P3} as follows.
\begin{theorem}  \label{T2.2}
\emph{\cite{H-P3}}
 Suppose that the coupling strength and the initial data satisfy
\begin{align}
\begin{aligned} \label{New-1}
& ||T_j^{in}||_F = 1, \quad  \kappa_{0} > 0, \quad   \sum_{i_*\neq0}\kappa_{i_*}   < \frac{\kappa_{0}}{4 ||T_c^{in}||^2_F}, \\
& 0< {\mathcal D}(T^{in})<\frac{\kappa_{0}- 4 \sum_{i_*\neq0}\kappa_{i_*} ||T_c^{in}||^2_F}{2\kappa_0},
\end{aligned}
\end{align}
and let $\{T_i \}$ be a solution to \eqref{A-1}. Then, one has 
\[
{\mathcal D}(T(t)) \leq C_1 e^{-\left( \kappa_{0}- 4 \sum_{i_*\neq0}\kappa_{i_*} ||T_c^{in}||^2_F  \right)t}, \quad t \geq 0.
\]
\end{theorem}
\begin{proof} The detailed proof can be found in Theorem 5.1 in \cite{H-P3}. In fact, the key idea is to derive a Gronwall type differential inequality for ensemble diameter ${\mathcal D}(T)$:
\[   \left|{d\over{dt}}{\mathcal D}(T)+\kappa_{0} {\mathcal D}(T) \right| 
\leq 2\kappa_{0} {\mathcal D}(T)^2+ \Big( \sum_{i_*\neq0}\kappa_{i_*} \Big) {\mathcal D}(T). \]
This certainly implies the desired decay estimate. 
\end{proof}
\begin{remark}
Note that the zero convergence of ensemble diameter does not imply the convergence of each state $T_j$ as $t \to \infty$, moreover, we do not have a complete picture for initial configuration and coupling strengths which do not satisfy the conditions \eqref{New-1}. The noticeable differences of main results in this paper is to relax the conditions \eqref{New-1} and derive detailed asymptotic structure resulting from the relaxation dynamics. It turns out that under rather general condition on the coupling strengths and initial data, we can show that two asymptotic patterns emerge from initial configuration, either one-state concentration and bi-polar concentration. This asymptotic picture coincides with that of the Kuramoto model for identical oscillators (see  Theorem \ref{T4.1} and Theorem \ref{T4.2}).
\end{remark}
\section{Reshaping of tensor contractions} \label{sec:3}
\setcounter{equation}{0}
In this section, we study how tensor contractions can be understood as matrix multiplications for matrices (rank-2 tensors) with a larger size. Recall that any $d_1 \times d_2$ complex matrix can be understood as a complex vector in $\bbc^{d_1 \times d_2}$ by juxtaposing column vectors, i.e., rank-2 tensor can be viewed as a rank-1 tensor with a larger size. Likewise rank-3 tensor can be rewritten as rank-1 and rank-2 tensors by juxtaposing some of indices. As a preliminary step in this direction, we consider how  tensor contractions between rank-2 tensors can recast as products of two matrices: for rank-2 tensors with the same size, 
\begin{align}
\begin{aligned} \label{C-0}
& [A]_{\alpha \gamma} [B]_{\gamma \beta} = [AB]_{\alpha \beta}, \quad [A]_{\alpha \gamma} [\bar{B}]_{\beta \gamma} =  [A]_{\alpha \gamma} [B^\dagger]_{\gamma \beta} = [AB^\dagger]_{\alpha \beta}, \\
& [A]_{\gamma \alpha} [\bar{B}]_{\gamma \beta} =  [A^t]_{\alpha \gamma} [{\bar B}]_{\gamma \beta} = [A^t \bar{B}]_{\alpha \beta},
\end{aligned}
\end{align}
where $A^{\dagger}$ is the Hermitian conjugate of $A$. 
\vspace{0.2cm}

Next, we extend the above special cases \eqref{C-0} to the tensor contraction for general rank-$m$ tensors. More precisely,  let $A$ and $B$ be rank-$m$ tensors in ${\mathcal T}_{m}({\mathbb C})$. Then, for $i_*\in\{0, 1\}^m$, we are interested in rewriting the following quantity:
\begin{equation} \label{C-1}
[A]_{\alpha_{*i_*}}[\bar{B}]_{\alpha_{*1}}
\end{equation}
as a component of rank-2 tensor. 

\subsection{Motivation} \label{sec:3.1} As a preliminary start-up, consider the simplest cases, rank-2 and rank-3 tensors. \newline

\noindent $\bullet$ (Rank-2 tensors): Suppose that $A, B \in {\mathcal T}_2(d_1, d_2; \bbc)$ and $i_* \in \{0, 1\}^2$. Then, there are four cases:  \newline

$\diamond$~Case A.1:  For $i_* = (0, 1)$, one has 
\[  [A]_{\alpha_{10} \alpha_{21}} [\bar{B}]_{\alpha_{11} \alpha_{21}} = [A]_{\alpha_{10} \alpha_{21}} [B^\dagger]_{\alpha_{21} \alpha_{11}} = [AB^\dagger]_{\alpha_{10} \alpha_{11}}, \]
where $B^\dagger$ is the Hermitian conjugate of $B$. In the sequel, this case will be treated as a {\it standard form}.

\vspace{0.2cm}

$\diamond$~Case A.2:  For $i_* = (0, 0)$, one has 
\begin{equation} \label{C-2}
[A]_{\alpha_{10} \alpha_{20}} [\bar{B}]_{\alpha_{11} \alpha_{21}} = [A]_{\alpha_{10} \alpha_{20}} [{\bar B}]_{\alpha_{11} \alpha_{21}}.
\end{equation}
Now we consider the natural embedding $J$ of $\bbc^{d_1 \times d_2} \to \bbc^{d_1d_2}$ and regrouping map:
\begin{equation*} \label{C-2-1}
 {\mathcal R}_{0}: \{ 1, \cdots, d_1 \} \times \{1, \cdots, d_2 \} \to \{1, \cdots, d_1 \times d_2 \}, \quad [JA]_{{\mathcal R}_0(i, j)}=[A]_{ij}.
 \end{equation*}
 For definiteness, we set 
 \[ {\mathcal R}_0(i, j) =  i + d_1(j-1). \] 
 Then, one has
\[ J(A) \in {\mathcal T}_1(d_1d_2;\bbc), \quad J(B) \in {\mathcal T}_1(d_1d_2;\bbc). \]
In this way, the relation \eqref{C-2} can be understood as 
\[
 [A]_{\alpha_{11} \alpha_{20}} [\bar{B}]_{\alpha_{11} \alpha_{21}} = [J(A)]_{{\mathcal R}_{0}(\alpha_{10}, \alpha_{20})} [J(B)]_{{\mathcal R}_{0}( \alpha_{11}, \alpha_{21})} = [J(A) \otimes J(B)]_{{\mathcal R}_{0}(\alpha_{10}, \alpha_{20}) {\mathcal R}_{0}( \alpha_{11}, \alpha_{21})}. 
\]

\vspace{0.2cm}

$\diamond$~Case A.3: For $i_* = (1, 0)$, one has 
\[  [A]_{\alpha_{11} \alpha_{20}} [\bar{B}]_{\alpha_{11} \alpha_{21}} = [A^{t}]_{\alpha_{20} \alpha_{11}} [\bar{B}]_{\alpha_{11} \alpha_{21}} = [A^t \bar{B}]_{\alpha_{20}\alpha_{21}}. \]

\vspace{0.2cm}

$\diamond$~Case A.4: For $i_* = (1, 1)$, one has 
\[  [A]_{\alpha_{11} \alpha_{21}} [\bar{B}]_{\alpha_{11} \alpha_{21}} =  [A]_{\alpha_{11} \alpha_{21}} [{\bar B}]_{\alpha_{11} \alpha_{21}} = \mbox{tr}[AB],   \]
which is rank-0 tensor. Note that rank-0 tensor can also be understood as a rank-2 tensor in ${\mathcal T}_1(1, 1;\bbc)$. 
\vspace{0.5cm}

\noindent $\bullet$ (Rank-3 tensors):~Let $A$ and $B$ rank-3 tensors in ${\mathcal T}_3(d_1, d_2, d_3; \bbc)$ and $i_* \in \{0, 1\}^3$. Then, there are eight cases for $i_*$. Among them, we only consider the following two cases for a motivation. The other cases can be treated similarly. 
\[  i_* = (0,0, 1), \quad i_* = (1,1,0).     \]

$\diamond$~Case B.1: For $i_* = (0,0, 1)$, the relation \eqref{C-1} becomes
\begin{equation} \label{C-3}
[A]_{\alpha_{10}\alpha_{20}\alpha_{31}}[\bar{B}]_{\alpha_{11}\alpha_{21}\alpha_{31}}.
\end{equation}
Similar to Case A.2, we can embed $A$ and $B$ into ${\mathcal T}_3(d_1,d_2, d_3; \bbc)$, and denote by $J(A), J(B) \in {\mathcal T}_2(d_1d_2, d_3; \bbc)$ by abuse of notation.  Then, relation \eqref{C-3} can be rewritten as follows.
\begin{align*}
\begin{aligned}\label{C-4}
[A]_{\alpha_{10}\alpha_{20}\alpha_{31}}[\bar{B}]_{\alpha_{11}\alpha_{21}\alpha_{31}} &= [J(A)]_{{\mathcal R}_{0}(\alpha_{10},\alpha_{20}) \alpha_{31}} [\overline{J(B)}]_{{\mathcal R}_{0}(\alpha_{11}, \alpha_{21})\alpha_{31}} \\
&= [J(A)J(B)^\dagger ]_{ {\mathcal R}_{0}(\alpha_{10},\alpha_{20}) {\mathcal R}_{0}(\alpha_{11}, \alpha_{21})}.
\end{aligned}
\end{align*}

\vspace{0.2cm}

$\diamond$~Case B.2: For $i_*=(1,1,0)$, \eqref{C-1} becomes 
\begin{align*}
\begin{aligned}
[A]_{\alpha_{11}\alpha_{21}\alpha_{30}}[\bar{B}]_{\alpha_{11}\alpha_{21}\alpha_{31}} &= [J(A)]_{{\mathcal R}_{1}(\alpha_{11},\alpha_{21})\alpha_{30}}[\overline{J(B)}]_{{\mathcal R}_{1}(\alpha_{11},\alpha_{21})\alpha_{31}}  \\
&=  [J(A)^t]_{\alpha_{30}{\mathcal R}_{1}(\alpha_{11},\alpha_{21})}  [\overline{J(B)}]_{{\mathcal R}_{1}(\alpha_{11},\alpha_{21}) \alpha_{31}}  \\
& = [J(A)^t \overline{J(B)}]_{\alpha_{30} \alpha_{31}}.
\end{aligned}
\end{align*}
The other cases can be treated similarly.

\subsection{Rank-$m$ tensors} For $i_*\in\{0, 1\}^m$,  consider the quantity:
\begin{equation} \label{C-4-1}
[A]_{\alpha_{*i_*}}[\bar{B}]_{\alpha_{*1}} = [A]_{\alpha_{1i_1} \alpha_{2 i_2} \cdots \alpha_{mi_m} }[\bar{B}]_{\alpha_{11} \alpha_{21} \cdots \alpha_{m1}}. 
\end{equation}
In what follows, we reshape the tensor contraction \eqref{C-4-1} as a matrix product between matrices with bigger sizes into two steps: \newline
\begin{itemize}
\item
Step A:~For a standard form index $i_*$, we use the grouping of indices as in Case A.1 and Case B.1 and canonical embeddings to rewrite \eqref{C-4-1} as a matrix product. 

\vspace{0.2cm}

\item
Step B:~For a non-standard form index $i_*$, we first use a permutation map $P_{i_*}$ to transform the given index into a standard form. After we reshaping the index into a standard form, we can use Step A to rewrite \eqref{C-4-1} as a matrix product. 
\end{itemize}

\vspace{0.2cm}

\subsubsection{Step A} Suppose that $i_*$ is given by a standard form:
\begin{equation} \label{C-4-2}
i_*=\left(i_{a_1},  \cdots, i_{a_n}, i_{b_1}, \cdots, i_{b_{m-n}}\right)=(\underbrace{0, 0, \cdots, 0}_{n\mbox{ times}}, \underbrace{1, 1, \cdots, 1}_{m-n\mbox{ times}}) =: i_{n, m-n},
\end{equation}
For a fixed $i_*$, we define rearrangement maps ${\mathcal R}_0$ and ${\mathcal R}_1$ associated with $i_*$ as in rank-2 tensors. \newline

\noindent Define a bijective map ${\mathcal R}_0$:
\begin{align}
\begin{aligned} \label{C-5}
& {\mathcal R}_0: \{1, 2, \cdots, d_1\}\times \cdots\times\{1, 2, \cdots, d_n\}\rightarrow\{1, 2, \cdots, d_1 d_2 \cdots d_n\}  \\
&  \hspace{1cm} \mbox{by} \quad  {\mathcal R}_0(i_1, \cdots, i_n) = i_1 + d_1 (i_2 - 1) + \cdots + d_1 \cdots d_{n-1} (i_n - 1).
\end{aligned}
\end{align}
Similarly, we also define a bijective map ${\mathcal R}_1$:
\begin{align}
\begin{aligned} \label{C-6}
&{\mathcal R}_1: \{1, 2, \cdots, d_{n+1}\} \times\cdots\times\{1, 2, \cdots, d_m\}\rightarrow\{1, 2, \cdots, d_{n+1}\cdots d_m\}, \\
& \hspace{1cm} \mbox{by} \quad {\mathcal R}_1(i_{n+1}, \cdots, i_{m}) = i_{n+1} + d_{n+1} (i_{n+2} - 1) + \cdots + d_{n+1} \cdots d_{m-1} (i_m - 1).
\end{aligned}
\end{align}
Now, for $i_*$ with \eqref{C-4-2}, consider \eqref{C-4-1}:
\begin{align*}
\begin{aligned}
{[A]}_{\alpha_{*i_*}}[\bar{B}]_{\alpha_{*1}} &= [A]_{\alpha_{10} \cdots \alpha_{n0}  \alpha_{(n+1)1} \cdots \alpha_{m1}} [\bar{B}]_{\alpha_{11} \cdots \alpha_{n1} \alpha_{(n+1)1} \cdots \alpha_{m1}} \\
& =  [J(A)]_{{\mathcal R}_0(\alpha_{10}, \cdots, \alpha_{n0}) {\mathcal R}_1(\alpha_{(n+1)1}, \cdots, \alpha_{m1})} [\overline{J(B)}]_{{\mathcal R}_0(\alpha_{11}, \cdots, \alpha_{n1}) {\mathcal R}_1(\alpha_{(n+1)1}, \cdots, \alpha_{m1})}\\
& =  [J(A)]_{{\mathcal R}_0(\alpha_{10}, \cdots, \alpha_{n0}) {\mathcal R}_1(\alpha_{(n+1)1}, \cdots, \alpha_{m1})} [J(B)^\dagger]_{{\mathcal R}_1(\alpha_{(n+1)1}, \cdots, \alpha_{m1}){\mathcal R}_0(\alpha_{11}, \cdots, \alpha_{n1}) }  \\
& = [J(A) J(B)^\dagger]_{{\mathcal R}_0(\alpha_{10}, \cdots, \alpha_{n0}) {\mathcal R}_0(\alpha_{11}, \cdots, \alpha_{n1})}.
\end{aligned}
\end{align*}

\subsubsection{Step B}  Let $i_* \in \{0, 1\}^m$ and assume that $i_*$ is not in a standard form:
\[
i_* \not = i_{n, m-n} \quad \mbox{for some}\quad 0\leq n \leq m.
\]
In this case, we rebel a general form into a standard form using the method of permutation. \newline

Suppose that the number of $0$ in $i_*$ is $n$ and the number of $1$ in $i_*$ is $m-n$. Then, we can define a partition $\{{\mathcal I}_0, {\mathcal I}_1 \}$ of $\{1, \cdots, N \}$:
\[
\begin{cases}
\displaystyle {\mathcal I}_0 :=\{a_1, a_2, \cdots, a_n\}=\{l: i_l=0\}, \quad \mathcal{I}_1 :=\{b_1, b_2, \cdots, b_{m-n}\}=\{l:i_l=1\}, \cr
\displaystyle a_1<a_2<\cdots<a_n,\quad b_1<b_2<\cdots<b_{m-n}, \quad \mathcal{I}_0 \cup\mathcal{I}_1 =\{1, 2, \cdots, m\}, \quad \mathcal{I}_0 \cap\mathcal{I}_1=\phi,
\end{cases}
\]
and we can define a permutation $P_{i_*}$ on the set $\{1, \cdots, m \}$:
\begin{align*}
&P_{i_*}(1)=a_1, \qquad P_{i_*}(2)=a_2,\qquad \cdots , \qquad P_{i_*}(n)=a_n,\\
&P_{i_*}(n+1)=b_1, \quad P_{i_*}(n+2)=b_2, \quad \cdots, \quad P_{i_*}(m)=b_{m-n}.
\end{align*}
Then, it is easy to see
\begin{align*}
i_{P_{i_*}(*)}&=\left(i_{P_{i_*}(1)}, i_{P_{i_*}(2)}, \cdots,  i_{P_{i_*}(m)}\right) =\left(i_{a_1},  \cdots, i_{a_n}, i_{b_1}, \cdots, i_{b_{m-n}}\right) \\
&=(\underbrace{0, 0, \cdots, 0}_{n\mbox{ times}}, \underbrace{1, 1, \cdots, 1}_{m-n\mbox{ times}})=i_{n, m-n}.
\end{align*}
Let $T \in {\mathcal T}_m(d_1, \cdots, d_m; \bbc)$ be a rank-$m$ tensor. Then, we define the rank-$m$ tensor ${T}^{P_{i_*}} \in {\mathcal T}_m(d_{P_{i_*}(1)}, \cdots, d_{P_{i_*}(m)}; \bbc)$ to satisfy
\begin{equation} \label{C-7}
[T]_{\alpha_1\alpha_2\cdots\alpha_m}=[T^{P_{i_*}}]_{\alpha_{a_1}\cdots\alpha_{a_n}\alpha_{b_1}\cdots\alpha_{b_{m-n}}}, \quad \mbox{i.e.,} \quad 
[T^{P_{i_*}}]_{\alpha_{*}} :=[T]_{\alpha_{P_{i_*}^{-1}(*)}}.
\end{equation}
Now we will define the matrix $M^{i_*}(T)$ from the tensor $T$.  As in \eqref{C-5} and \eqref{C-6}, we define rearrangement maps:
\begin{align*}
&{\mathcal R}_0^{i_*}: \{1, 2, \cdots, d_{P_{i_*}(1)}\}\times\{1, 2, \cdots, d_{P_{i_*}(2)}\}\times\cdots\times\{1, 2, \cdots, d_{P_{i_*}(n)}\}\\
&\hspace{6cm}\longrightarrow\{1, 2, \cdots, d_{P_{i_*}(1)}\times d_{P_{i_*}(2)}\times \cdots\times d_{P_{i_*}(n)}\}, \\
&{\mathcal R}_1^{i_*}: \{1, 2, \cdots, d_{P_{i_*}(n+1)}\}\times\{1, 2, \cdots, d_{P_{i_*}(n+2)}\}\times\cdots\times\{1, 2, \cdots, d_{P_{i_*}(m)}\}\\
&\hspace{5cm}\longrightarrow\{1, 2, \cdots, d_{P_{i_*}(n+1)}\times d_{P_{i_*}(n+2)}\times \cdots\times d_{P_{i_*}(m)}\}.
\end{align*}

\vspace{0.2cm}

\begin{definition} \label{D3.1}
Let $T \in {\mathcal T}(d_1, \cdots, d_m;\bbc)$ and $i_* \in \{0, 1\}^m$. Then, the rank-2 tensor $M^{i_*}(T)$ reshaped from $T$ is given as follows.
\[
[M^{i_*}(T)]_{{\mathcal R}_0^{i_*}(\alpha_{a_1}, \alpha_{a_2}, \cdots, \alpha_{a_n}) {\mathcal R}_1^{i_*}(\alpha_{b_1}, \alpha_{b_2},\cdots,\alpha_{b_{m-n}})}:=[T^{P_{i_*}}]_{\alpha_1\alpha_2\cdots\alpha_m} = [T]_{\alpha_{P_{i_*}^{-1}(*)}}.
\]
\end{definition}
\begin{remark}
Note that $M^{i_*}$ is a bijective linear map from $({\mathcal T}_m(d_1, \cdots, d_m;\bbc), \| \cdot \|_F)$ to $({\mathcal T}_m(  d_{P_{i_*}(1)} \cdot d_{P_{i_*}(2)} \cdots  {P_{i_*}(n)}, d_{P_{i_*}(n+1)} \cdots d_{P_{i_*}(m)}; \bbc), \| \cdot \|_F)$. Moreover, it is an isometry, when the tensor space and matrix space are equipped with Frobenius norms.
\end{remark}

\vspace{0.5cm}

Now, we return to \eqref{C-4-1}.  We use \eqref{C-7} to see
\begin{align*}
\begin{aligned}
{[A]}_{\alpha_{*i_*}}[\bar{B}]_{\alpha_{*1}}&=[A^{P_{i_*}}]_{\alpha_{P_{i_*}(*)i_{P_{i_*}(*)}}} [\overline{B^{P_{i_*}}}]_{\alpha_{P_{i_*}(*)1}} \\
&=[A^{P_{i_*}}]_{\alpha_{P_{i_*}(*) i_{n, m-n}}} [\overline{B^{P_{i_*}}}]_{\alpha_{P_{i_*}(*)1}}\\
&=[M^{i_*}(A)]_{{\mathcal R}_0^{i_*}(\alpha_{P_{i_*}(1)0}\cdots\alpha_{P_{i_*}(n)0}) {\mathcal R}_1^{i_*}(\alpha_{P_{i_*}(n+1)1}\cdots\alpha_{P_{i_*}(m)1})} \\
& \times [\overline{M^{i_*}(B)}]_{{\mathcal R}_0^{i_*}(\alpha_{P_{i_*}(1)1}\cdots\alpha_{P_{i_*}(n)1}) {\mathcal R}_1^{i_*}(\alpha_{P_{i_*}(n+1)1}\cdots\alpha_{P_{i_*}(m)1})}\\
&=[M^{i_*}(A)M^{i_*}(B)^\dagger]_{{\mathcal R}_0^{i_*}(\alpha_{P_{i_*}(1)0}\cdots\alpha_{P_{i_*}(n)0}) {\mathcal R}_0^{i_*}(\alpha_{P_{i_*}(1)1}\cdots\alpha_{P_{i_*}(n)1})},
\end{aligned}
\end{align*}
i.e., we have
\[
[A]_{\alpha_{*i_*}} [\bar{B}]_{\alpha_{*1}} =[M^{i_*}(A)M^{i_*}(B)^\dagger]_{{\mathcal R}_0^{i_*}(\alpha_{P_{i_*}(1)0}\cdots\alpha_{P_{i_*}(n)0}) {\mathcal R}_0^{i_*}(\alpha_{P_{i_*}(1)1}\cdots\alpha_{P_{i_*}(n)1})}.
\]
In next lemma, we study how the cubic contraction can be reshaped as the product of three matrices. 
\begin{lemma} \label{L3.1}
Let $A, B, C$ and $D$ be tensors in ${\mathcal T}(d_1, \cdots, d_m; \bbc)$ such that 
\begin{equation} \label{C-8}
[D]_{\alpha_{*0}}=[A]_{\alpha_{*i_*}}[\bar{B}]_{\alpha_{*1}}[C]_{\alpha_{*(1-i_*)}},
\end{equation}
where $i_* \in \{0, 1 \}^m$.  Then one has
\[ M^{i_*}(D)=M^{i_*}(A)M^{i_*}(B)^\dagger M^{i_*}(C). \]
\end{lemma}
\begin{proof} It follows from Definition \ref{D3.1} that 
\begin{align}
\begin{aligned} \label{C-9}
&[A]_{\alpha_{*i_*}} [\bar{B}]_{\alpha_{*1}} [C]_{\alpha_{*(1-i_*)}} \\
&\hspace{1cm} =[A^{P_{i_*}}]_{\alpha_{P_{i_*}(*)i_{P_{i_*}(*)}}} [\overline{B^{P_{i_*}}}]_{\alpha_{P_{i_*}(*)1}}[C^{P_{i_*}}]_{\alpha_{P_{i_*}(*)(1-i_{P_{i_*}(*)})}}\\
&\hspace{1cm}=[A^{P_{i_*}}]_{\alpha_{P_{i_*}(*)i_{n, m-n}}} [\overline{B^{P_{i_*}}}]_{\alpha_{P_{i_*}(*)1}}[C^{P_{i_*}}]_{\alpha_{P_{i_*}(*)(1-i_{n, m-n})}}\\
&\hspace{1cm}=[M^{i_*}(A)]_{{\mathcal R}_0^{i_*}(\alpha_{P_{i_*}(1)0}\cdots\alpha_{P_{i_*}(n)0}) {\mathcal R}_1^{i_*}(\alpha_{P_{i_*}(n+1)1}\cdots\alpha_{P_{i_*}(m)1})} \\
&\hspace{1cm} \times [\overline{M^{i_*}(B)}]_{{\mathcal R}_0^{i_*}(\alpha_{P_{i_*}(1)1}\cdots\alpha_{P_{i_*}(n)1}) {\mathcal R}_1^{i_*}(\alpha_{P_{i_*}(n+1)1}\cdots\alpha_{P_{i_*}(m)1})} \\
&\hspace{1cm}\times [M^{i_*}(C)]_{{\mathcal R}_0^{i_*}(\alpha_{P_{i_*}(1)1}\cdots\alpha_{P_{i_*}(n)1}) {\mathcal R}_1^{i_*}(\alpha_{P_{i_*}(n+1)0}\cdots\alpha_{P_{i_*}(m)0})}\\
&\hspace{1cm}=[M^{i_*}(A) M^{i_*}(B)^{\dagger} M^{i_*}(C)]_{{\mathcal R}_0^{i_*}(\alpha_{P_{i_*}(1)0}\cdots\alpha_{P_{i_*}(n)0}) {\mathcal R}_1^{i_*}(\alpha_{P_{i_*}(n+1)0}\cdots\alpha_{P_{i_*}(m)0})}.
\end{aligned}
\end{align}
On the other hand, one has

\begin{align}
\begin{aligned} \label{C-10}
&[M^{i_*}(D)]_{{\mathcal R}_0^{i_*}(\alpha_{P_{i_*}(1)0}\cdots\alpha_{P_{i_*}(n)0}) {\mathcal R}_1^{i_*}(\alpha_{P_{i_*}(n+1)0}\cdots\alpha_{P_{i_*}(m)0})} \\
& \hspace{0.5cm} =[D^{P_{i_*}}]_{\alpha_{P_{i_*}(1)0}\cdots\alpha_{P_{i_*}(m)0}} =[D]_{\alpha_{10}\cdots\alpha_{m0}}=[D]_{\alpha_{*0}} =[A]_{\alpha_{*i_*}}[\bar{B}]_{\alpha_{*1}}[C]_{\alpha_{*(1-i_*)}},
\end{aligned}
\end{align}
where we used the relation \eqref{C-8}. 
\newline

Finally we combine \eqref{C-9} and \eqref{C-10} to get
\begin{align*}
\begin{aligned}
&[M^{i_*}(A) (M^{i_*}(B))^{\dagger} M^{i_*}(C)]_{{\mathcal R}_0^{i_*}(\alpha_{P_{i_*}(1)0}\cdots\alpha_{P_{i_*}(n)0}) {\mathcal R}_1^{i_*}(\alpha_{P_{i_*}(n+1)0}\cdots\alpha_{P_{i_*}(m)0})} \\
& \hspace{3cm} = [M^{i_*}(D)]_{{\mathcal R}_0^{i_*}(\alpha_{P_{i_*}(1)0}\cdots\alpha_{P_{i_*}(n)0}) {\mathcal R}_1^{i_*}(\alpha_{P_{i_*}(n+1)0}\cdots\alpha_{P_{i_*}(m)0})}.
\end{aligned}
\end{align*}
This yields the desired result:
\[
M^{i_*}(D)=M^{i_*}(A) (M^{i_*}(B))^\dagger M^{i_*}(C).
\]
\end{proof}

\subsection{Fundamental Lohe tensor models} \label{sec:3.3}
In this subsection, we consider a situation in which the Lohe tensor model \eqref{A-1} can be reshaped into a generalized Lohe matrix model proposed in \cite{H-P1}. In fact, we will show that if the Lohe tensor model has only one cubic interaction, it can be transformed into a generalized Lohe matrix model using the reshaping of tensor contractions studied in previous subsection. Thus, the Lohe tensor model with only one coupling term will be called a ``{\it  fundamental Lohe tensor model}". 
\begin{definition} \label{D3.2}
\begin{enumerate}
\item
For each multi-index $i_* \in \{0,1\}^{m}$, we call the coupling term in \eqref{B-1} involved with the coupling strength $\kappa_{i_*}$:
\[
[T_c]_{\alpha_{*i_*}}[\bar{T_j}]_{\alpha_{*1}}[T_j]_{\alpha_{*(1-i_*)}}-[T_j]_{\alpha_{*i_*}}[\bar{T_c}]_{\alpha_{*1}}[T_j]_{\alpha_{*(1-i_*)}},
\]
as {\it the fundamental Lohe coupling associated with $i_*$}.  

\vspace{0.2cm}

\item
If  the Lohe tensor model contains only one fundamental coupling term and natural frequencies $A$ satisfies condition for solution splitting property:
\[
[e^{-At}]_{\alpha_{*0}\beta_{*0}}=[e^{-At}]_{\alpha_{*i_*}\beta_{*i_*}}[e^{At}]_{\beta_{*1}\alpha_{*1}}[e^{-At}]_{\alpha_{*(1-i_*)}\beta_{*(1-i_*)}},
\]
then the subsystem
\begin{equation} \label{C-8}
\frac{d}{dt}[T_j]_{\alpha_{*0}}={\kappa_{i_*}}([T_c]_{\alpha_{*i_*}}[\bar{T}_j]_{\alpha_{*1}}[T_j]_{\alpha_{*(1-i_*)}}-[T_j]_{\alpha_{*i_*}}[\bar{T}_c]_{\alpha_{*1}}[T_j]_{\alpha_{*(1-i_*)}}), 
\end{equation}
is called \textbf{the fundamental Lohe tensor model related with $i_*$.} 
\end{enumerate}
\end{definition}
\begin{remark}
Note that the Lohe tensor model \eqref{B-1} can be viewed as a linear combination of the free flow and the fundamental Lohe coupling related with $i_*$. Thus, in some sense, the fundamental Lohe tensor model will play a building block for the Lohe tensor model \eqref{B-1}.
\end{remark}
Note that the fundamental Lohe tensor model associated with $i_*$ is given by
\begin{align}
\begin{cases}\label{C-9}
\displaystyle\frac{d}{dt}[T_j]_{\alpha_{*0}}={\kappa_{i_*}}([T_c]_{\alpha_{*i_*}}[\bar{T}_j]_{\alpha_{*1}}[T_j]_{\alpha_{*(1-i_*)}}-[T_j]_{\alpha_{*i_*}}[\bar{T}_c]_{\alpha_{*1}}[T_j]_{\alpha_{*(1-i_*)}}),\\
T_j(0)=T_j^{in}.
\end{cases}
\end{align}
Then we  use  Lemma \ref{L3.1} to rewrite \eqref{C-9} as a generalized Lohe matrix model:
\begin{align}\label{C-10}
\begin{cases}
\displaystyle\frac{d}{dt}M^{i_*}(T_j)={\kappa_{i_*}}(M^{i_*}(T_c)M^{i_*}(T_j)^{\dagger} M^{i_*}(T_j)-M^{i_*}(T_j)M^{i_*}(T_c)^{\dagger} M^{i_*}(T_j)),\\
T_j(0)=T_j^{in}.
\end{cases}
\end{align}
Then we can obtain following lemma.

\begin{proposition} 
Let $\{T_j\}_{j=1}^N$ be a solution of the system \eqref{C-9}. Then
\[
M^{i_*}(T_j)^\dagger M^{i_*}(T_j)
\]
is a conserved quantity for each $j = 1, \cdots, N$.
\end{proposition}
\begin{proof}
Instead of the system \eqref{C-9}, we can use the system \eqref{C-10}. Then we can obtain
\begin{align*}
\begin{aligned}
&\frac{d}{dt}\left(M^{i_*}(T_j)^\dagger M^{i_*}(T_j)\right) \\
& \hspace{0.5cm} =\left(\frac{d}{dt}M^{i_*}(T_j)^\dagger \right)M^{i_*}(T_j)+M^{i_*}(T_j)^\dagger \left(\frac{d}{dt}M^{i_*}(T_j)\right)\\
&  \hspace{0.5cm} =\kappa_{i_*}M^{i_*}(T_j)^\dagger (M^{i_*}(T_j)M^{i_*}(T_c)^\dagger -M^{i_*}(T_c)M^{i_*}(T_j)^\dagger )M^{i_*}(T_j)\\
& \hspace{0.5cm}+\kappa_{i_*}M^{i_*}(T_j)^\dagger (M^{i_*}(T_c) M^{i_*}(T_j)^{\dagger} -M^{i_*}(T_j) M^{i_*}(T_c)^{\dagger})M^{i_*}(T_j)=0.
\end{aligned}
\end{align*}
Hence the quanitty
\[
M^{i_*}(T_j)^\dagger M^{i_*}(T_j)
\]
is conserved for each $j = 1, \cdots, N$.
\end{proof}

\section{Emergent dynamics} \label{sec:4}
\setcounter{equation}{0}
In this section, we present an improved aggregation estimate compared to \cite{H-P3} by introducing an order parameter and nonlinear functional measuring the degree of aggregation and study their time-evolution along the Lohe tensor flow.
\subsection{Lyapunov functionals} \label{sec:4.1}  
Let $\{T_j \}$ be a time-dependent Lohe tensor ensemble with $\|T_j \|_F = 1$. Then, we define a centroid and its Frobenius norm which can play the role of an order parameter $R(T)$ for aggregation: 
\begin{equation*} \label{D-0}
T_c := \frac{1}{N} \sum_{j=1} T_j, \quad R(T) := \| T_c \|_F.
\end{equation*}
Moreover, we also introduce variance and maximal radius functionals:
\begin{equation} \label{D-1}
{\mathcal V}(T) :=\frac{1}{N}\sum_{j=1}^N\langle{T_j -T_c, T_j-T_c}\rangle_F, \quad  {\mathcal F}(T) :=\max_{1 \leq j \leq N}||T_j -T_c||_F.
\end{equation}
In next lemma, we study the relationship between two functionals defined in \eqref{D-1}.
\begin{lemma} \label{L4.1}
Let $\{T_j \}$ be an ensemble of rank-$m$ tensors whose dynamics is governed by \eqref{B-1} with $\|T_j \|_F = 1$. Then, the following assertions hold.
\begin{enumerate}
\item
The variance functional and order parameter are related by the following relation:
\[ {\mathcal V}(T) = 1 - |R(T)|^2. \]
\item
The functionals ${\mathcal D}(T), {\mathcal F}(T)$ and ${\mathcal V}(T)$ satisfy
\[
\frac{1}{N} {\mathcal F}(T)^2 \leq {\mathcal V}(T)\leq {\mathcal D}(T)^2.
\]
\end{enumerate}
\end{lemma}
\begin{proof} 
(i) We use definitions in \eqref{D-1} and $\|T_j \|_F = 1$ to get 
\[ {\mathcal V}(T) =\left(\frac{1}{N}\sum_{j=1}^N||T_j ||^2_F\right)-||T_c||^2_F  = 1 - |R(T)|^2. \]
(ii) We use the following simple relations 
\[ T_j - T_c = \frac{1}{N} \sum_{k=1}^{N} (T_j - T_k) \quad \mbox{and} \quad  \|T_j - T_k \|_F \leq {\mathcal D}(T) \]
to see
\[
\frac{1}{N}{\mathcal F}(T)^2\leq \frac{1}{N}\sum_{j=1}^N||T_j -T_c||^2_F \leq\frac{1}{N}\sum_{i=1}^N {\mathcal D}(T)^2= {\mathcal D}(T)^2.
\]
\end{proof}

Next, we study the time-evolution of the functional ${\mathcal V}(T)$.
\begin{proposition} \label{P4.1}
Let $\{T_i \}$ be an ensemble of rank-$m$ tensors whose dynamics is governed by \eqref{B-1}. Then, the following assertions hold. \newline
\begin{enumerate}
\item
The variance functional is non-increasing along the Lohe tensor flow:
\begin{equation} \label{D-2}
\frac{d}{dt} {\mathcal V}(T)=-\frac{1}{N}\sum_{j=1}^N\sum_{i_* \in \{0, 1\}^m}\kappa_{i_*}||M^{i_*}(T_c)M^{i_*}(T_j)^{\dagger}-M^{i_*}(T_j)M^{i_*}(T_c)^{\dagger}||_F^2 \leq 0.
\end{equation}
\item
For all $j = 1, \cdots, N$ and $i_* \in \{0, 1\}^m$ with $\kappa_{i_*}\neq0$, one has
\[ \lim_{t \to \infty} ||M^{i_*}(T_c)M^{i_*}(T_j)^{\dagger}-M^{i_*}(T_j)M^{i_*}(T_c)^{\dagger}||_F = 0. \]
\end{enumerate}
\end{proposition}
\begin{proof} (i)~Note that $T_j$ and $T_c$ satisfy
\begin{align*}
\begin{aligned} 
\frac{d}{dt}[T_j]_{\alpha_{*0}} &=\sum_{i_*}{\kappa_{i_*}}([T_c]_{\alpha_{*i_*}}[\bar{T}_j]_{\alpha_{*1}}[T_j]_{\alpha_{*(1-i_*)}}-[T_j]_{\alpha_{*i_*}}[\bar{T}_c]_{\alpha_{*1}}[T_j]_{\alpha_{*(1-i_*)}}), \\
\frac{d}{dt}[T_c]_{\alpha_{*0}} &=\frac{1}{N}\sum_{j=1}^N\sum_{i_*}{\kappa_{i_*}}([T_c]_{\alpha_{*i_*}}[\bar{T}_j]_{\alpha_{*1}}[T_j]_{\alpha_{*(1-i_*)}}-[T_j]_{\alpha_{*i_*}}[\bar{T}_c]_{\alpha_{*1}}[T_j]_{\alpha_{*(1-i_*)}}).
\end{aligned}
\end{align*}
These imply
\begin{align*}
\begin{aligned}
\frac{d}{dt}[T_j-T_c]_{\alpha_{*0}} &=-\frac{1}{N}\sum_{j=1}^N\sum_{i_*}{\kappa_{i_*}}([T_c]_{\alpha_{*i_*}}[\bar{T}_j]_{\alpha_{*1}}[T_j]_{\alpha_{*(1-i_*)}}-[T_c]_{\alpha_{*i_*}}[\bar{T}_j]_{\alpha_{*1}}[T_j]_{\alpha_{*(1-i_*)}}\\
&\hspace{1cm}-[T_j]_{\alpha_{*i_*}}[\bar{T}_c]_{\alpha_{*1}}[T_j]_{\alpha_{*(1-i_*)}}+[T_j]_{\alpha_{*i_*}}[\bar{T}_c]_{\alpha_{*1}}[T_j]_{\alpha_{*(1-i_*)}}).
\end{aligned}
\end{align*}
Now we differentiate ${\mathcal V}(T)$ in \eqref{D-1} to see
\begin{align}
\begin{aligned} \label{D-4}
& \frac{d}{dt} {\mathcal V}(T)  =\frac{1}{N}\sum_{i=1}^N\frac{d}{dt}||T_i-T_c||^2 =\frac{1}{N}\sum_{i=1}^N\frac{d}{dt}\left([T_i-T_c]_{\alpha_{*0}}[\overline{T_i-T_c}]_{\alpha_{*0}}\right)\\
&\hspace{0.5cm}  =\frac{1}{N}\sum_{i=1}^N\left(\frac{d}{dt}[T_i-T_c]_{\alpha_{*0}}\cdot[\overline{T_i-T_c}]_{\alpha_{*0}}+\mbox{c.c}\right)\\
&\hspace{0.5cm}  =-\frac{1}{N^2}\sum_{i, j}\sum_{i_*}{\kappa_{i_*}}\big(([T_c]_{\alpha_{*i_*}}[\bar{T}_j]_{\alpha_{*1}}[T_j]_{\alpha_{*(1-i_*)}}-[T_c]_{\alpha_{*i_*}}[\bar{T}_i]_{\alpha_{*1}}[T_i]_{\alpha_{*(1-i_*)}}\\
&\hspace{1cm}-[T_j]_{\alpha_{*i_*}}[\bar{T}_c]_{\alpha_{*1}}[T_j]_{\alpha_{*(1-i_*)}}+[T_i]_{\alpha_{*i_*}}[\bar{T}_c]_{\alpha_{*1}}[T_i]_{\alpha_{*(1-i_*)}})[\overline{T_i-T_c}]_{\alpha_{*0}}+\mbox{c.c}\big)\\
&\hspace{0.5cm}  =-\frac{1}{N}\sum_{i=1}^N\sum_{i_*}{\kappa_{i_*}}\big((-[T_c]_{\alpha_{*i_*}}[\bar{T}_i]_{\alpha_{*1}}[T_i]_{\alpha_{*(1-i_*)}}+[T_i]_{\alpha_{*i_*}}[\bar{T}_c]_{\alpha_{*1}}[T_i]_{\alpha_{*(1-i_*)}})[\overline{T_i-T_c}]_{\alpha_{*0}}+\mbox{c.c}\big)\\
&\hspace{0.5cm} =-\frac{1}{N}\sum_{i=1}^N\sum_{i_*}{\kappa_{i_*}}\big(-[T_c]_{\alpha_{*i_*}}[\bar{T}_i]_{\alpha_{*1}}[T_i]_{\alpha_{*(1-i_*)}}[\bar{T}_i]_{\alpha_{*0}}+[T_c]_{\alpha_{*i_*}}[\bar{T}_i]_{\alpha_{*1}}[T_i]_{\alpha_{*(1-i_*)}}[\bar{T}_c]_{\alpha_{*0}}\\
&\hspace{1cm} +[T_i]_{\alpha_{*i_*}}[\bar{T}_c]_{\alpha_{*1}}[T_i]_{\alpha_{*(1-i_*)}}[\bar{T}_i]_{\alpha_{*0}}-[T_i]_{\alpha_{*i_*}}[\bar{T}_c]_{\alpha_{*1}}[T_i]_{\alpha_{*(1-i_*)}}[\bar{T}_c]_{\alpha_{*0}}+\mbox{c.c}\big).
\end{aligned}
\end{align}
On the other hand, since
\[ [A]_{\alpha_{*i_*}}[\bar{B}]_{\alpha_{*1}}[C]_{\alpha_{*(1-i_*)}}[\bar{D}]_{\alpha_{*0}}=[A]_{\alpha_{*1}}[\bar{B}]_{\alpha_{*i_*}}[C]_{\alpha_{*0}}[\bar{D}]_{\alpha_{*(1-i_*)}}=\overline{[B]_{\alpha_{*i_*}}[\bar{A}]_{\alpha_{*1}}[D]_{\alpha_{*(1-i_*)}}[\bar{C}]_{\alpha_{*0}}}, \]
one has 
\begin{align}
\begin{aligned} \label{D-5}
&[T_c]_{\alpha_{*i_*}}[\bar{T}_i]_{\alpha_{*1}}[T_i]_{\alpha_{*(1-i_*)}}[\bar{T}_i]_{\alpha_{*0}}=\overline{[T_i]_{\alpha_{*i_*}}[\bar{T}_c]_{\alpha_{*1}}[T_i]_{\alpha_{*(1-i_*)}}[\bar{T}_i]_{\alpha_{*0}}}, \\
& [T_i]_{\alpha_{*i_*}}[\bar{T}_c]_{\alpha_{*1}}[T_i]_{\alpha_{*(1-i_*)}}[\bar{T}_i]_{\alpha_{*0}}=\overline{[T_i]_{\alpha_{*i_*}}[\bar{T}_c]_{\alpha_{*1}}[T_i]_{\alpha_{*(1-i_*)}}[\bar{T}_i]_{\alpha_{*0}}}.
\end{aligned}
\end{align}
We combine \eqref{D-4} and \eqref{D-5} to get
\begin{align}
\begin{aligned} \label{D-6}
&\frac{d}{dt}{\mathcal V}(T) =-\frac{1}{N}\sum_{i=1}^N\sum_{i_*}{\kappa_{i_*}}\big([T_c]_{\alpha_{*i_*}}[\bar{T}_i]_{\alpha_{*1}}[T_i]_{\alpha_{*(1-i_*)}}[\bar{T}_c]_{\alpha_{*0}}-[T_i]_{\alpha_{*i_*}}[\bar{T}_c]_{\alpha_{*1}}[T_i]_{\alpha_{*(1-i_*)}}[\bar{T}_c]_{\alpha_{*0}}+\mbox{c.c}\big)\\
&\hspace{0.5cm} =-\frac{1}{N}\sum_{i=1}^N\sum_{i_*}{\kappa_{i_*}}\big([T_c]_{\alpha_{*i_*}}[\bar{T}_i]_{\alpha_{*1}}[T_i]_{\alpha_{*(1-i_*)}}[\bar{T}_c]_{\alpha_{*0}}-[T_i]_{\alpha_{*i_*}}[\bar{T}_c]_{\alpha_{*1}}[T_i]_{\alpha_{*(1-i_*)}}[\bar{T}_c]_{\alpha_{*0}}\\
&\hspace{4.5cm} +[T_i]_{\alpha_{*i_*}}[\bar{T}_c]_{\alpha_{*1}}[T_c]_{\alpha_{*(1-i_*)}}[\bar{T}_i]_{\alpha_{*0}}-[T_c]_{\alpha_{*i_*}}[\bar{T}_c]_{\alpha_{*1}}[T_c]_{\alpha_{*(1-i_*)}}[\bar{T}_i]_{\alpha_{*0}}\big)\\
&\hspace{0.5cm} =\frac{1}{N}\sum_{i=1}^N\sum_{i_*}\kappa_{i_*}([T_c]_{\alpha_{*i_*}}[\bar{T}_i]_{\alpha_{*1}}-[T_i]_{\alpha_{*i_*}}[\bar{T}_c]_{\alpha_{*1}})([T_c]_{\alpha_{*(1-i_*)}}[\bar{T}_i]_{\alpha_{*0}}-[T_i]_{\alpha_{*(1-i_*)}}[\bar{T}_c]_{\alpha_{*0}})\\
&\hspace{0.5cm}=\frac{1}{N}\sum_{i=1}^N\sum_{i_*}\kappa_{i_*}(\overline{[\bar{T}_c]_{\alpha_{*i_*}}[T_i]_{\alpha_{*1}}-[\bar{T}_i]_{\alpha_{*i_*}}[T_c]_{\alpha_{*1}}})([T_c]_{\alpha_{*(1-i_*)}}[\bar{T}_i]_{\alpha_{*0}}-[T_i]_{\alpha_{*(1-i_*)}}[\bar{T}_c]_{\alpha_{*0}}).
\end{aligned}
\end{align}
We can proceed the estimate \eqref{D-6} further using delicate tensor contractions. However, it would be better to transform the R.H.S. into a combination of matrix products to simplify the analysis. For this, note that 
\begin{align}
\begin{aligned} \label{D-7}
&[\bar{T}_c]_{\alpha_{*i_*}}[T_i]_{\alpha_{*1}}-[\bar{T}_i]_{\alpha_{*i_*}}[T_c]_{\alpha_{*1}}\\
&=[\overline{M^{i_*}(T_c)M^{i_*}(T_j)^{\dagger}-M^{i_*}(T_j)M^{i_*}(T_c)^\dagger}]_{{\mathcal R}_0^{i_*}(\alpha_{P_{i_*}(1)0}\cdots\alpha_{P_{i_*}(n)0}) {\mathcal R}_0^{i_*}(\alpha_{P_{i_*}(1)1}\cdots\alpha_{P_{i_*}(n)1})}\\
&=[(M^{i_*}(T_c)M^{i_*}(T_j)^\dagger-M^{i_*}(T_j)M^{i_*}(T_c)^\dagger)^\dagger]_{{\mathcal R}_0^{i_*}(\alpha_{P_{i_*}(1)1}\cdots\alpha_{P_{i_*}(n)1}) {\mathcal R}_0^{i_*}(\alpha_{P_{i_*}(1)0}\cdots\alpha_{P_{i_*}(n)0})}\\
&=-[M^{i_*}(T_c)M^{i_*}(T_j)^\dagger-M^{i_*}(T_j)M^{i_*}(T_c)^\dagger]_{{\mathcal R}_0^{i_*}(\alpha_{P_{i_*}(1)1}\cdots\alpha_{P_{i_*}(n)1}) {\mathcal R}_0^{i_*}(\alpha_{P_{i_*}(1)0}\cdots\alpha_{P_{i_*}(n)0})}.
\end{aligned}
\end{align}
Similarly, 
\begin{align}
\begin{aligned} \label{D-8}
&[T_c]_{\alpha_{*(1-i_*)}}[\bar{T}_i]_{\alpha_{*0}}-[T_i]_{\alpha_{*(1-i_*)}}[\bar{T}_c]_{\alpha_{*0}}\\
&\hspace{0.5cm} =[M^{i_*}(T_c)M^{i_*}(T_j)^\dagger-M^{i_*}(T_j)M^{i_*}(T_c)^\dagger]_{{\mathcal R}_0^{i_*}(\alpha_{P_{i_*}(1)1}\cdots\alpha_{P_{i_*}(n)1}) {\mathcal R}_0^{i_*}(\alpha_{P_{i_*}(1)0}\cdots\alpha_{P_{i_*}(n)0})}.
\end{aligned}
\end{align}
In \eqref{D-6}, we use \eqref{D-7} and \eqref{D-8} to get
\[
\frac{d}{dt}{\mathcal V}(T)=-\frac{1}{N}\sum_{j=1}^N\sum_{i_*}\kappa_{i_*}||M^{i_*}(T_c)M^{i_*}(T_j)^\dagger -M^{i_*}(T_j)M^{i_*}(T_c)^\dagger||_F^2.
\]

\vspace{0.5cm}

\noindent (ii)~We use similar arguments for the Lohe sphere model  and Lohe matrix model as in \cite{H-P1}, i.e., since 
\[ \mathcal{V}(T)\geq 0 \quad \mbox{and} \quad \frac{d}{dt}\mathcal{V}(T)\leq0, \]
we can see that 
\[ \exists~\lim_{t \to \infty} \mathcal{V}(T(t)). \]
Moreover, one can easily show the boundedness of $\frac{d^2}{dt^2}V(T(t))$. Hence, it follows from Barbalat's lemma that  
\[
\lim_{t \to \infty} ||M^{i_*}(T_c)M^{i_*}(T_j)^\dagger-M^{i_*}(T_j)M^{i_*}(T_c)^\dagger||_F = 0,
\]
for all $j$ and $i_*$ with $\kappa_{i_*}\neq0$. 
\end{proof}
\subsection{Applications to low-rank models} \label{sec:4.2} Next we consider three explicit low-rank models, and explain how the result of Proposition \ref{P4.1} can be reinterpreted for low-rank models. 

\subsubsection{The Kuramoto model} Let $A$ and $B$ be scalars in $\mathbb{C}$. In this case, we have 
\[ m=0 \quad \mbox{and} \quad i_*=\phi. \]
So one has
\[
M^{i_*}(A)=A,\quad M^{i_*}(B)=B,
\]
where $M^{i_*}(A)$ and $M^{i_*}(B)$ are $1\times 1$ matrix. Thus, they can be considered as scalar. Then we can obtain
\[
M^{i_*}(A)M^{i_*}(B)^\dagger-M^{i_*}(B)M^{i_*}(A)^\dagger=A\bar{B}-B\bar{A}.
\]
We set 
\[ T_j=e^{i\theta_j}, \quad j = 1, \cdots, N, \qquad T_c= Re^{{\mathrm i}\phi}, \]
where $R$ is an order parameter. Finally, we can obtain
\begin{equation} \label{D-6}
M^{i_*}(T_c)M^{i_*}(T_j)^\dagger-M^{i_*}(T_j)M^{i_*}(T_c)^\dagger=Re^{\mathrm{i}(\phi-\theta_j)}-Re^{\mathrm{i}(\theta_j-\phi)}=2R\mathrm{i}\sin(\phi-\theta_j).
\end{equation}
If we substitute the above relation \eqref{D-6} into equation \eqref{D-2}, we obtain
\[
\frac{d}{dt} {\mathcal V}(T)=-\frac{1}{N}\sum_{i=1}^N4R^2\sin^2(\phi-\theta_i). 
\]
\subsubsection{The Lohe sphere model} Let $A$ and $B$ be vectors in $\mathbb{C}^{d_1}$. In vector case, we have $m=1$ so there are two possible $i_*$:
\[
i_*=(0) \quad \mbox{and} \quad i_*=(1).
\]

\vspace{0.5cm}

\noindent $\bullet$~Case 1.1~$(i_*=(0))$: In this case, one has
\[
[M^{i_*}(A)]_{{\mathcal R}_0^{i_*}(\alpha_1)1}=[A^{P_{i_*}}]_{\alpha_1}=[A]_{\alpha_1},\quad[M^{i_*}(B)]_{{\mathcal R}_0^{i_*}(\alpha_1)1}=[B]_{\alpha_1}.
\]
This yields
\[
[M^{i_*}(A)M^{i_*}(B)^\dagger]_{{\mathcal R}_0^{i_*}(\alpha_1) {\mathcal R}_0^{i_*}(\alpha_2)}=[A]_{\alpha_1}[\bar{B}]_{\alpha_2}.
\]
Finally, one has
\[
[M^{i_*}(A)M^{i_*}(B)^\dagger-M^{i_*}(B)M^{i_*}(A)^\dagger]_{{\mathcal R}_0^{i_*}(\alpha_1) {\mathcal R}_0^{i_*}(\alpha_2)}=[A]_{\alpha_1}[\bar{B}]_{\alpha_2}-[B]_{\alpha_1}[\bar{A}]_{\alpha_2}.
\]
Thus, we have
\begin{align*}
\begin{aligned}
&||M^{i_*}(A)M^{i_*}(B)^\dagger-M^{i_*}(B)M^{i_*}(A)^\dagger||_F^2\\
& \hspace{1cm} =([A]_{\alpha_1}[\bar{B}]_{\alpha_2}-[B]_{\alpha_1}[\bar{A}]_{\alpha_2})([\bar{A}]_{\alpha_1}{[B]_{\alpha_2}}-[\bar{B}]_{\alpha_1}{[A]_{\alpha_2}})\\
& \hspace{1cm}  =2||A||_F^2||B||_F^2-\langle{A, B}\rangle_F^2-\langle{B, A}\rangle_F^2=2||A||_F^2||B||_F^2-2\mathrm{Re}(\langle{A, B}\rangle_F^2).
\end{aligned}
\end{align*}

\vspace{0.5cm}

\noindent $\bullet$~Case 1.2~$(i_*=(1))$: In this case, one has
\[
[M^{i_*}(A)]_{1\mathcal{R}_1^{i_*}(\alpha_1)}=[A^{P_{i_*}}]_{\alpha_1}=[A]_{\alpha_1},\quad[M^{i_*}(B)]_{1\mathcal{R}_1^{i_*}(\alpha_1)}=[B]_{\alpha_1}.
\]
So we can obtain
\[
[M^{i_*}(A)M^{i_*}(B)^\dagger]_{11}=[A]_{\alpha_1}[\bar{B}]_{\alpha_1}=\langle{B, A}\rangle_F.
\]
Also we can obtain
\[
[M^{i_*}(A)M^{i_*}(B)^\dagger-M^{i_*}(B)M^{i_*}(B)^\dagger]_{11}=\langle{B, A}\rangle_F-\langle{A, B}\rangle_F=-2\mathrm{i}\mathrm{Im}(\langle{A, B}\rangle_F).
\]
Finally we can obtain
\[
||M^{i_*}(A)M^{i_*}(B)^\dagger-M^{i_*}(B)M^{i_*}(A)^\dagger||_F^2\\=4(\mathrm{Im}(\langle{A, B}\rangle))^2.
\]

If we combine above two cases, we can obtain
\[
\frac{d}{dt}{\mathcal V}(T)=-\frac{\kappa_0}{N}\sum_{i=1}^N \Big (2||T_j||_F^2||T_c||_F^2-2\mathrm{Re}(\langle{T_j, T_c}\rangle_F^2) \Big)-\frac{\kappa_1}{N}\sum_{i=1}^N4(\mathrm{Im}(\langle{T_i, T_c}\rangle))^2.
\]
Note that the above result is same with the complex Lohe sphere model.

\subsubsection{The Lohe matrix model}
Let A and B be matrices in $\mathbb{C}^{d_1\times d_2}$. Since $m=2$, we have following cases:
\[
i_*=(0, 0),\quad (0, 1),\quad (1, 0),\quad (1, 1).
\]
The cases with $i_*=(0, 0)$ and $i_*=(1, 1)$ are very similar to the case of vector when $i_*=(0)$ and $i_*=(1)$ respectively. Hence, we consider only the following cases:
\[ i_*=(0, 1) \quad \mbox{and} \quad i_*=(1, 0). \]

\vspace{0.5cm}

\noindent $\bullet$~Case 2.1~$(i_*=(0,1))$: In this case, one has
\[
[M^{i_*}(A)]_{{\mathcal R}_0^{i_*}(\alpha_1) {\mathcal R}_1^{i_*}(\alpha_2)}=[A^{P_{i_*}}]_{\alpha_1\alpha_2}=[A]_{\alpha_1\alpha_2},\qquad [M^{i_*}(B)]_{{\mathcal R}_0^{i_*}(\alpha_1) {\mathcal R}_1^{i_*}(\alpha_2)}=[B]_{\alpha_1\alpha_2}.
\]
Thus, one has
\[
[M^{i_*}(A)M^{i_*}(B)^\dagger]_{\mathcal{R}_0^{i_*}(\alpha_1)\mathcal{R}_0^{i_*}(\beta_1)}=[AB^\dagger]_{\alpha_1\beta_1}.
\]
Finally we can obtain that
\[
[M^{i_*}(A)M^{i_*}(B)^\dagger-M^{i_*}(B)M^{i_*}(A)^\dagger]_{\mathcal{R}_0^{i_*}(\alpha_1)\mathcal{R}_0^{i_*}(\beta_1)}=[AB^\dagger-BA^\dagger]_{\alpha_1\beta_1},
\]
and from this we have
\[
||M^{i_*}(A)M^{i_*}(B)^\dagger-M^{i_*}(B)M^{i_*}(A)^\dagger||_F^2=||AB^\dagger-BA^\dagger||_F^2.
\]

\vspace{0.5cm}

\noindent $\bullet$~Case 2.2~$(i_*=(1,0))$:~In this case, we can use similar argument with the case when $i_*=(0, 1)$. Then we can obtain
\[
||M^{i_*}(A)M^{i_*}(B)^\dagger-M^{i_*}(B)M^{i_*}(A)^\dagger||_F^2=||A^\dagger B-B^\dagger A||_F^2.
\]

If we combine above two results, we can obtain follows:
\[
\frac{d}{dt} {\mathcal V}(T)=-\frac{\kappa_{01}}{N}\sum_{i=1}^N||T_cT_j^\dagger-T_jT_c^\dagger||_F^2-\frac{\kappa_{10}}{N}\sum_{i=1}^N||T_c^\dagger T_j-T_j^\dagger T_c||_F^2.
\]
\begin{remark}
For the case $m=2$, if $T_i$ is unitary, we have
\[
||T_cT_i^\dagger-T_iT_c^\dagger||_2^2=||T_i^\dagger T_c-T_c^\dagger T_i||_2^2.
\]
So the terms involving with $\kappa_{01}$ and $\kappa_{10}$ has no different, however if $T_i$ is not unitary, then the two coupling terms will induce different effects. 
\end{remark}

\subsection{Complete aggregation} \label{sec:4.3}
In this subsection, we study the emergence of two distinguished states(completely aggregated state and bi-polar state) in the relaxation process. Before we present our asymptotic dynamics, we recall Barbalat's lemma as follows.
\begin{lemma}  \label{L4.2}
\emph{\cite{Ba}}
\emph{(i)}~Suppose that  a real-valued function $f: [0, \infty) \to \bbr$ is uniformly continuous and it satisfies
\[ \lim_{t \to \infty} \int_0^t f(s)d s \quad \textup{exists}. \]
Then, $f$ tends to zero as $t \to \infty$:
\[ \lim_{t \to \infty} f(t) = 0. \]
\emph{(ii)}~Suppose that a real-valued function $f: [0, \infty) \to \bbr$ is continuously differentiable, and $\lim_{t \to \infty} f(t) = \alpha \in \bbr$. If $f^{\prime}$ is uniformly continuous, then 
\[ \lim_{t \to \infty} f^{\prime}(t)  = 0. \]
\end{lemma}
Our first asymptotic result deals with the emergence two distinguished states.
\begin{theorem}\label{T4.1}
Suppose that the coupling strength $\kappa_{i_*}$ satisfies
\begin{equation*} \label{D-10-1}
 \kappa_{00\cdots0}>0 \quad \mbox{and} \quad  \kappa_{i_*}\geq 0,  \quad \forall~i_* \neq (0, \cdots, 0),
\end{equation*} 
and let $\{T_i\}$ be a solution of system \eqref{B-1}. Then, the following assertions hold.
\begin{enumerate}
\item
For an index $i_*$ such that $\kappa_{i_*} > 0$, one has 
\[ \lim_{t \to \infty} ||M^{i_*}(T_c)M^{i_*}(T_j)^\dagger-M^{i_*}(T_j)M^{i_*}(T_c)^\dagger||_F^2 = 0. \]
\item
There exists a vector ${\bf a} = (a_2, \cdots, a_N) \in \{-1, 1 \}^{N-1}$ such that 
\begin{equation} \label{D-12-0}
\lim_{t \to \infty}(T_i(t) - a_i T_1(t)) = 0, \quad i = 2, \cdots, N.
\end{equation}
\end{enumerate}
\end{theorem}
\begin{proof}
\noindent (i)~It follows from Proposition \ref{P4.1} that 
\begin{equation} \label{D-12}
\frac{d}{dt} {\mathcal V}(T)=-\frac{1}{N}\sum_{j=1}^N\sum_{i_*}\kappa_{i_*}||M^{i_*}(T_c)M^{i_*}(T_j)^\dagger-M^{i_*}(T_j)M^{i_*}(T_c)^\dagger||_F^2 \leq 0.
\end{equation}
Since ${\mathcal V}(T)$ is non-increasing and ${\mathcal V}(T) \geq 0$, one has
\begin{equation*} \label{D-13}
\exists~\lim_{t \to \infty} {\mathcal V}(T). 
\end{equation*} 
Next, we show that $\frac{d^2}{dt^2}{\mathcal V}(T)$ is uniformly bounded. For this, we recall that the Lohe tensor model:
\begin{equation}  \label{D-14}
\displaystyle \frac{d}{dt}[T_j]_{\alpha_{*0}} = \sum_{i_* \in \{0,1\}^m}{\kappa_{i_*}}([T_c]_{\alpha_{*i_*}}[\bar{T}_j]_{\alpha_{*1}}[T_j]_{\alpha_{*(1-i_*)}}-[T_j]_{\alpha_{*i_*}}[\bar{T}_c]_{\alpha_{*1}}[T_j]_{\alpha_{*(1-i_*)}}).
\end{equation}
Then, \eqref{D-14} and $\| T_j \|_F = 1$ imply
\begin{equation} \label{D-15}
\sup_{0 \leq t < \infty} \max_{1 \leq i \leq N} \Big(  \| T_i(t) \|_F  +  \Big \| \frac{dT_i(t)}{dt} \Big \|_F \Big) < \infty. 
\end{equation}
We also differentiate \eqref{D-12} with respect to $t$ and use \eqref{D-15} to get 
\begin{equation} \label{D-16}
\sup_{0 \leq t < \infty} \max_{1 \leq i \leq N}  \Big \| \frac{d^2T_i(t)}{dt^2} \Big \|_F < \infty. 
\end{equation}
We differentiate the definition of ${\mathcal V}(T)$ in \eqref{D-1} twice with respect to $t$ and obtain the uniform boundedness of $\frac{d^2}{dt^2} {\mathcal V}(T)$:
\begin{equation*} \label{D-17}
 \sup_{0 \leq t < \infty} \Big| \frac{d^2}{dt^2} {\mathcal V}(T) \Big| < \infty. 
\end{equation*} 
Thus, $\frac{d}{dt} {\mathcal V}(T)$ is uniformly continuous. Now, we can apply Babalat's lemma (Lemma \ref{L4.2}) using the relations \eqref{D-15} and \eqref{D-16} to get 
\[ \frac{d}{dt} {\mathcal V}(T) =-\frac{1}{N}\sum_{j=1}^N\sum_{i_*}\kappa_{i_*}||M^{i_*}(T_c)M^{i_*}(T_j)^\dagger-M^{i_*}(T_j)M^{i_*}(T_c)^\dagger||_F^2 \to 0, \quad \mbox{as $t \to \infty$}.
\]

\vspace{0.5cm}

\noindent (ii) Since $\kappa_{00\cdots0}>0$, for $i_*=(0, 0, \cdots, 0)$, one has
\[
 ||M^{i_*}(T_c)M^{i_*}(T_j)^*-M^{i_*}(T_j)M^{i_*}(T_c)^*||_F^2=\sum_{\alpha_*,\beta_*}||[T_c]_{\alpha_*}[\bar{T}_j]_{\beta_*}-[T_j]_{\alpha_*}[\bar{T}_c]_{\beta_*}||_F^2.
\]
This yields
\begin{equation} \label{D-16-1}
[T_c]_{\alpha_*}[\bar{T}_j]_{\beta_*}-[T_j]_{\alpha_*}[\bar{T}_c]_{\beta_*}\rightarrow0,\quad \mbox{as $t \to \infty$}.
\end{equation}
for all $j,~\alpha_*$ and $\beta_*$. Then we have
\[
[T_c]_{\alpha_*}[\bar{T}_j]_{\beta_*}[T_j]_{\beta_*}-[T_j]_{\alpha_*}[\bar{T}_c]_{\beta_*}[T_j]_{\beta_*}=[T_c]_{\alpha_*}||T_j||_F^2-\langle{T_c, T_j}\rangle [T_j]_{\alpha_*}\rightarrow0,\quad \mbox{as $t \to \infty$},
\]
which means
\begin{equation} \label{D-17}
||T_j||_F^2T_c-\langle{T_c, T_j}\rangle T_j\rightarrow0, \quad \mbox{as $t \to \infty$}.
\end{equation}
If we define the function $c_j(t)$:
\begin{equation} \label{D-18}
c_j(t)=\frac{\langle{T_c(t), T_j(t)}\rangle_F}{||T_j^{in}||_F^2},
\end{equation}
it follows from \eqref{D-17} and \eqref{D-18}  that 
\begin{equation} \label{D-18-1}
T_c-c_jT_j\rightarrow 0, \quad \mbox{as $t \to \infty$}.
\end{equation}
The relations \eqref{D-16-1} and \eqref{D-18-1} imply
\begin{align}
\begin{aligned} \label{D-18-2}
&(c_j-\bar{c}_j)[T_j]_{\alpha_*}[\bar{T}_j]_{\beta_*} \\
& \hspace{0.5cm} =[T_c]_{\alpha_*}[\bar{T}_j]_{\beta_*}-[T_j]_{\alpha_*}[\bar{T}_c]_{\beta_*}  -[T_c-c_jT_j]_{\alpha_*}[\bar{T}_j]_{\beta_*} +[T_j]_{\alpha_*}[\bar{T}_c -\bar{c}_j\bar{T}_j]_{\beta_*} \rightarrow 0.
\end{aligned}
\end{align}
Since \eqref{D-18-2} holds for any $\alpha_*$ and $\beta_*$, $c_j-\bar{c}_j$ tends to zero asymptotically, i.e., 
\[ \lim_{t \to \infty}\mbox{Im}(c_j) = 0. \] 
If we define $r_j=\mathrm{Re}(c_j)$, then we have
\[
T_c-r_jT_j=(T_c-c_jT_j)+\mathrm{i}\mathrm{Im}(c_j)T_j\rightarrow0,\quad\mbox{as $t \to \infty$}.
\]
Now, we use the fact $||T_c||_F\rightarrow R_\infty$ to deduce 
\begin{align*}
\begin{aligned}
& \left(||T_c||_F-||T_c-r_jT_j||_F\right)^2 \\
& \hspace{0.5cm} \leq r_j^2||T_j||_F^2=||r_jT_j||_F^2=||T_c-(T_c-r_jT_j)||_F^2\leq\left(||T_c||_F+||T_c-r_jT_j||_F\right)^2.
\end{aligned}
\end{align*}
Hence,
\[  r_j^2(t)\rightarrow \frac{R_\infty^2}{||T_j||_F^2}=R_\infty^2\quad\mbox{as $t \to \infty$},
\]
where we used $||T_j||_F=1$ for all $j=1, 2, \cdots, N$. Therefore, one has
\[
r_j(t)\rightarrow \pm R_\infty\quad\mbox{as $t \to \infty$}.
\]
We define
\begin{equation} \label{D-18-3}
b_j =
\begin{cases}
1,\quad & r_j(t)\rightarrow R_\infty\quad\mbox{as $t \to \infty$},\\
-1,\quad & r_j(t)\rightarrow - R_\infty\quad\mbox{as $t \to \infty$}.
\end{cases}
\end{equation}
Then we have
\begin{equation} \label{D-19}
T_c-b_j R_\infty T_j=(T_c-r_j T_j)+(r_j-b_jR_\infty )T_j\rightarrow0,\quad\mbox{as $t \to \infty$}.
\end{equation}
Next, we set
\begin{equation} \label{D-20}
a_j :=b_1b_j,
\end{equation}
and we use \eqref{D-19} to see
\begin{align*}
\begin{aligned}
T_j-a_jT_1 &=T_j-b_1b_jT_1=b_j(b_jT_j-b_1T_1) \\
&= -\frac{b_j}{R_\infty}(T_c-b_jR_\infty T_j) + \frac{b_j}{R_\infty}(T_c-b_1 R_\infty T_1)\rightarrow0,\quad\mbox{as $t \to \infty$}.
\end{aligned}
\end{align*}
\end{proof}
\begin{remark}
\noindent 1. It follows from \eqref{D-18-3} and \eqref{D-20} that 
\[
\mbox{either} \quad a_j=1\quad\mbox{or}\quad -1.
\]
If $T_j-T_1$ converges to 0 then we have $a_j=1$ and if $T_j+T_1$ converges to 0 then we have $a_j=-1$. \newline

\noindent 2.~As a direct corollary of the second result of Theorem \ref{T4.1}, we can find the relation between $a_j$'s and asymptotic order parameter $R_\infty := \lim_{t \to \infty} R(t)$: By definition of $R$, triangle inequality and $\|T_j \|_F = 1$, we have
\begin{align*}
\begin{aligned}
&\left|\frac{1}{N}\sum_{k=1}^N||T_k-a_kT_1||_F-\frac{1}{N}\left|\sum_{k=1}^Na_k\right|\right| \leq R=\left\|\frac{1}{N}\sum_{k=1}^NT_k\right\|_F \\
& \hspace{0.5cm} =\left\|\frac{1}{N}\sum_{k=1}^N[(T_k-a_jT_1)+a_kT_1]  \right \|_F \leq\frac{1}{N}\sum_{k=1}^N||T_k-a_kT_1||_F+\frac{1}{N}\left|\sum_{k=1}^Na_k\right|,
\end{aligned}
\end{align*}
i.e., 
\[  \left|\frac{1}{N}\sum_{k=1}^N||T_k-a_kT_1||_F-\frac{1}{N}\left|\sum_{k=1}^Na_k\right|\right| \leq R \leq \frac{1}{N}\sum_{k=1}^N||T_k-a_kT_1||_F+\frac{1}{N}\left|\sum_{k=1}^Na_k\right|. \]
Letting $t \to \infty$ and we use \eqref{D-12-0} to get 
\[
R_\infty=\left|\frac{1}{N}\sum_{k=1}^Na_k\right|.
\]
On the other hand, it follows from \eqref{D-18-3} that 
\[
R_\infty=\frac{1}{N}\sum_{k=1}^Nb_k.
\]

\noindent 3.~For the Kuramoto model (the Lohe tensor model for rank-0 tensors), emergence of complete phase synchronization or bi-polar configuration has been studied in \cite{H-K-R}. Thus, our work is a high-dimensional extension of the earlier work on the Kuramoto model.
\end{remark}

\vspace{0.5cm}

Before we present the complete aggregation, we study an elementary lemma.
\begin{lemma} \label{L4.3}
Let $\{T_i \}$ be an ensemble of tensors such that 
\begin{equation} \label{D-21} 
T_1=\cdots=T_n=T,\quad T_{n+1}=\cdots=T_N=-T, \quad \|T \|_F = 1,
\end{equation}
for some $0\leq n \leq N$. Then, one has
\[ ||T_c||_F=\left|1-\frac{2n}{N}\right|. \]
\end{lemma}
\begin{proof} By definition of $T_c$ and \eqref{D-21}, one has
\[
T_c= \frac{1}{N} \sum_{j=1}^{N} T_j = \frac{1}{N} \Big( \sum_{j=1}^{n} T_j + \sum_{j=n+1}^{N} T_j \Big) = 
 \frac{1}{N} \Big( n T - (N-n) T \Big) = \left(\frac{N-2n}{N}\right)T.
\]
This and $\|T \|_F$ yield
\[  ||T_c||_F= \Big| 1-\frac{2n}{N} \Big|. \]
\end{proof}
Now, we provide a sufficient condition for complete aggregation. 
\begin{theorem} \label{T4.2}
Suppose that the coupling strength $\kappa_{i_*}$ and the initial data satisfy
\[ \kappa_{00\cdots0}>0 \quad \mbox{and} \quad  \kappa_{i_*}\geq 0,  \quad \forall~i_* \neq (0, \cdots, 0), \quad  R^{in} >1-\frac{2}{N}, \]
and let $\{T_i\}$ be a solution of system \eqref{B-1}. Then, complete state aggregation(one-point concentration) occurs asymptotically:
\[  \lim_{t \to \infty} \|T_i(t) - T_1(t) \|_F = 0, \quad i = 2, \cdots, N. \]
\end{theorem}
\begin{proof} We use Proposition \ref{P4.1} and Lemma \ref{L4.1} to see the monotonicity of $R$:
\[ R(t) \geq R^{in}, \quad t \geq 0. \]
Then, this and assumption on initial data imply
\begin{equation} \label{D-22}
 R(t) \geq R^{in} >  1-\frac{2}{N}, \quad t \geq 0. 
\end{equation} 
Suppose that bi-polar state emerges: for some $n \leq [N/2]$, one has
\begin{align*}
\begin{aligned} 
& \lim_{t \to \infty} \| T_j(t) -T(t) \|_F = 0, \quad  1 \leq j \leq n, \\
&  \lim_{t \to \infty} \| T_j(t) -(-T(t)) \|_F = 0, \quad  n+1 \leq j \leq N,
\end{aligned}
\end{align*}
where $\| T(t) \|_F = 1$. Then, it follows from Lemma \ref{L4.3} that 
\[  \lim_{t \to \infty} ||T_c(t)||_F=1-\frac{2n}{N},   \]
which is clearly contradictory to \eqref{D-22}.  Hence, we have the complete state aggregation.
\end{proof}

\section{conclusion} \label{sec:5}
\setcounter{equation}{0} 
In this paper, we have studied emergent dynamics of the Lohe tensor model with the same free flow. The Lohe tensor model describes aggregate dynamics of the ensemble of tensors with the same rank and size. Our Lohe tensor model includes the previously known aggregation models such as the Kuramoto model, the Lohe sphere model and the Lohe matrix model as special cases. Our main results can be summarized as follows. First, we  provide a sufficient and necessary condition for solution splitting property in the sense that the solution for the original Lohe tensor model can be rewritten as the composition of free flow and corresponding nonlinear Lohe tensor model. We also showed that the emergent behaviors emerge from any initial data, as long as the principle coupling strength is strictly positive. In a recent work \cite{H-P3}, the author studied emergent dynamics of the Lohe tensor model with more restricted conditions on the coupling strengths and initial data. We showed that any generic initial configuration tends to two distinguished states (completely aggregated state and bi-polar state) using a variance functional as a Lyapunov functional. For some conditions on initial data, we provided a sufficient condition leading to the completely aggregated state. In the current work, we only considered the Lohe tensor ensemble with the same free flow. However, when the free flows are heterogeneous, the formation of state-locked state is an open problem. So far, only weak aggregation estimate namely ``{\it practical aggregation estimates}" have been obtained in authors' previous work \cite{H-P3}. Hence, extension of the current work in the heterogeneous free flows will be a challenging problem which needs to be explored in a future work. 
\newline

\textit{\textbf{The data including papers that support the findings of this study are available from the corresponding author upon reasonable request.}}


\begin{thebibliography}{99}

\bibitem{A-B} Acebron, J. A., Bonilla, L. L., P\'{e}rez Vicente, C. J. P., Ritort, F. and Spigler, R.: \textit{The Kuramoto model: A simple paradigm for synchronization phenomena.} Rev. Mod. Phys. \textbf{77} (2005), 137-185.


\bibitem{A-B-F} Albi, G., Bellomo, N., Fermo, L., Ha, S.-Y., Kim, J., Pareschi, L., Poyato, D. and Soler, J.: \textit{Vehicular traffic, crowds and swarms: From kinetic theory and multiscale methods to applications and research perspectives.} Math. Models Methods Appl. Sci. {\bf 29} (2019), 1901-2005. 

%

\bibitem{B-C-M} Benedetto, D., Caglioti, E. and Montemagno, U.: \textit{On the complete phase
    synchronization for the Kuramoto model in the mean-field limit.}    Commun. Math. Sci. {\bf 13} (2015), 1775-1786.
    
\bibitem{Ba} Barb$\check{a}$lat, I.: \textit{Syst$\grave{e}$mes d$\acute{e}$quations diff$\acute{e}$rentielles dÕoscillations non Lin$\acute{e}$aires}. Rev. Math. Pures Appl. {\bf 4} (1959), 267-270.    
    
\bibitem{B-C} Bridgeman, J. C. and Chubb, C. T.: \textit{Hand-waving and interpretive dance: an Introductory course on tensor networks}.  J. Phys. A: Math. Theor. {\bf 50} (2017), 223001.

\bibitem{B-C-S} Bronski, J., Carty, T. and Simpson, S.: \textit{A matrix valued Kuramoto model.} J. Stat. Phys. {\bf 178} (2020), 595-624.


\bibitem{C-C-H} Chi, D., Choi, S.-H. and Ha, S.-Y.: \textit{Emergent behaviors of a holonomic particle system on a sphere.} J. Math. Phys. {\bf 55} (2014), 052703.

%
%
%


\bibitem{C-H5} Choi, S.-H. and Ha, S.-Y.: \textit{Complete entrainment of Lohe oscillators under attractive and repulsive couplings.} SIAM. J. App. Dyn. {\bf 13} (2013), 1417-1441.

\bibitem {C-H-J-K} Choi, Y., Ha, S.-Y., Jung, S. and Kim, Y.: \textit{Asymptotic formation and orbital stability of phase-locked states for the Kuramoto model.} Physica D {\bf 241} (2012), 735-754.





\bibitem{C-S} Chopra, N. and Spong, M. W.: \textit{On exponential synchronization of Kuramoto oscillators.} IEEE Trans. Automatic Control {\bf 54} (2009), 353-357.

\bibitem{D-F-M-T} Degond, P., Frouvelle, A., Merino-Aceituno, S., Trescases, A.:\textit{Quaternions in collective dynamics.} Multiscale Model. Simul. {\bf 16} (2018), 28--77.

\bibitem{D-F-M} Degond, P., Frouvelle, A., Merino-Aceituno, S.: \textit{A new flocking model through body attitude coordination.} Math. Models Methods Appl. Sci. {\bf 27} (2017), 1005--1049. 

\bibitem{De} DeVille, L.: \textit{Synchronization and stability for quantum Kuramoto.}  J. Stat. Phys. {\bf174} (2019), 160--187. 

\bibitem{D-X} Dong, J.-G. and Xue, X.: \textit{Synchronization analysis of Kuramoto oscillators.}  Commun. Math. Sci. {\bf 11} (2013), 465-480.

\bibitem{D-B1}  D\"{o}rfler, F. and Bullo, F.: \textit{Synchronization in complex networks of phase oscillators: A survey.} Automatica \textbf{50} (2014), 1539-1564.


\bibitem {D-B} D\"{o}rfler, F. and Bullo, F.: \textit{On the critical coupling for Kuramoto oscillators.} SIAM. J. Appl. Dyn. Syst. \textbf{10} (2011), 1070-1099.






\bibitem{H-K} Ha, S.-Y., Kim, D.: \textit{Emergent behavior of a second-order Lohe matrix model on the unitary group.} J. Stat. Phys. {\bf 175} (2019), 904-931.

\bibitem{H-K-R} Ha, S.-Y., Kim, H. W. and Ryoo, S. W.: \textit{Emergence of phase-locked states for the Kuramoto model in a large coupling regime.} Commun. Math. Sci. {\bf 14} (2016), 1073-1091.

\bibitem{H-K-P-Z} Ha, S.-Y., Ko, D., Park, J. and Zhang, X.: \textit{Collective synchronization of classical and quantum oscillators.} To appear in EMS Surveys in Mathematical Sciences {\bf 3} (2016), 209-267.

\bibitem{H-K-R} Ha, S.-Y., Ko, D. and Ryoo, S. W.: \textit{On the relaxation dynamics of Lohe oscillators on some Riemannian manifolds.} J. Stat. Phys. {\bf 172} (2018), 1427-1478.

\bibitem{H-L-X}  Ha, S.-Y., Li, Z. and Xue, X.: \textit{Formation of phase-locked states in a population of locally interacting Kuramoto oscillators.} J. Differential Equations {\bf 255} (2013), 3053-3070.

\bibitem{H-P1} Ha, S.-Y. and Park, H.: \textit{Emergent behaviors of the generalized Lohe matrix model.} Submitted.

\bibitem{H-P2} Ha, S.-Y. and Park, H.: \textit{From the Lohe tensor model to the complex Lohe sphere model and emergent dynamics.} Submitted.

\bibitem{H-P3} Ha, S.-Y. and Park, H.: \textit{Emergent behaviors of Lohe tensor flock:} J. Stat. Phys (2020).

\bibitem{H-R} Ha, S.-Y. and Ryoo, S.W.: \textit{On the emergence and orbital Stability of phase-locked states for the Lohe model} J. Stat. Phys {\bf 163} (2016), 411-439.








\bibitem {H-S} Ha, S.-Y. and Slemrod, M.: \textit{A fast-slow dynamical systems theory for the Kuramoto phase model.}   J. Diff. Eqs. \textbf{251}, 2685-2695 (2011).




\bibitem{Ku1} Kuramoto, Y.: \textit{Chemical oscillations, waves and turbulence.} Springer-Verlag, Berlin, 1984.

\bibitem{Ku2} Kuramoto, Y.: \textit{International symposium on mathematical problems in mathematical physics.} Lecture Notes Theor. Phys.  \textbf{30} (1975), 420.




\bibitem{Lo-0} Lohe, M. A.: \textit{Systems of matrix Riccati equations, linear fractional transformations, partial integrability and synchronization.} J. Math. Phys. {\bf 60} (2019), 072701.

\bibitem{Lo-1} Lohe, M. A.: \textit{Quantum synchronization over quantum networks.} J. Phys. A: Math. Theor. {\bf 43} (2010), 465301.

\bibitem{Lo-2} Lohe, M. A.: \textit{Non-abelian Kuramoto model and synchronization.} J. Phys. A: Math. Theor. {\bf 42} (2009), 395101.






\bibitem{M-T-G} Markdahl, J., Thunberg, J. and Gonçalves, J.: \textit{Almost global consensus on the n-sphere.} IEEE Trans. Automat. Control {\bf 63} (2018), 1664-1675. 

%
%

\bibitem{Or} Or\'{u}s, R.: \textit{A practical introduction to tensor networks: Matrix product states and projected entangled pair states}. Annals of Physics
{\bf 349} (2014), 117-158.


\bibitem{P-R} Pikovsky, A., Rosenblum, M. and Kurths, J.: \textit{Synchronization: A universal concept in
nonlinear sciences.} Cambridge University Press, Cambridge, 2001.


\bibitem{St} Strogatz, S. H.: \textit{From Kuramoto to Crawford: exploring the onset of synchronization in populations of coupled oscillators.} Physica D \textbf{143} (2000), 1-20.

\bibitem{T-M} Thunberg, J., Markdahl, J., Bernard, F. and Goncalves, J.: \textit{A lifting method for analyzing distributed synchronization on the unit sphere.} Automatica J. IFAC {\bf 96} (2018), 253-258. 

%

\bibitem{VZ} Vicsek, T. and Zefeiris, A.: \textit{Collective motion.} Phys. Rep. {\bf 517} (2012), 71-140.









\bibitem{Wi2} Winfree, A. T.: \textit{Biological rhythms and the behavior of populations of coupled oscillators.} J. Theor. Biol. \textbf{16} (1967), 15-42.

\bibitem{Wi1} Winfree, A. T.: \textit{The geometry of biological time.} Springer, New York, 1980.




\bibitem{Zhu}Zhu, J.: \textit{Synchronization of Kuramoto model in a high-dimensional linear space} Physics Letters A {\bf 377} (2013), 2939-2943.


\end{thebibliography}
\end{document}